\documentclass[12pt,a4,english,leqno]{article}
\usepackage[latin9]{inputenc}
\usepackage{geometry}
\usepackage{verbatim}
\usepackage{amsmath}
\usepackage{amsfonts}
\usepackage{amsthm}
\usepackage{amssymb}
\usepackage{mathtools}
\usepackage{graphicx}
\usepackage{longtable}
\usepackage{setspace}
\usepackage{comment}
\usepackage{upgreek}
\usepackage{xcolor}

\usepackage{lscape}
\usepackage{pdflscape}
\usepackage{empheq}

\usepackage[round,sort,authoryear]{natbib}
\usepackage[colorlinks=true,linkcolor=blue, citecolor=blue, urlcolor = blue]{hyperref}

\usepackage{titlesec}
\usepackage{babel}
\usepackage[toc,page]{appendix}
\usepackage{booktabs,caption}

\usepackage[labelfont=bf,textfont=md]{caption}
\usepackage[labelsep=space]{caption}

\usepackage{lmodern}
\usepackage[T1]{fontenc}

\DeclarePairedDelimiter{\floor}{\lfloor}{\rfloor}

\usepackage{bibentry}
\nobibliography*

\numberwithin{equation}{section}

\titleformat*{\section}{\large \bfseries}
\titleformat*{\subsection}{\normalsize \bfseries}
\titleformat*{\subsubsection}{\small \bfseries}

\newif\ifshow 
\showfalse 

\ifshow
  \includecomment{wrap}
\else
  \excludecomment{wrap} 
\fi

\theoremstyle{definition}
\newtheorem{theorem}{Theorem}

\newtheorem{assumption}{Assumption}

\newtheorem{lemma}{Lemma}
\newtheorem{example}{Example}
\newtheorem{proposition}{Proposition}
\newtheorem{remark}{Remark}
\newtheorem{corollary}{Corollary}

\newcommand\norm[1]{\left\lVert#1\right\rVert}

\begin{document}
\pagenumbering{roman}


\title{ {\Large \textbf{Robust Estimation in Network Vector Autoregression with Nonstationary Regressors}\thanks{\textbf{Article history:}  First draft: 21 April 2021. Second draft: 11 November 2023. The first draft of this paper was prepared after PhD Candidature was confirmed at the Department of Economics, University of Southampton. Relevant papers prepared during my PhD studies include \cite{katsouris2021optimal, katsouris2023estimating, katsouris2023limit, katsouris2023statistical}. The main ideas presented in the original draft are extended to Chapter 4 of my PhD thesis submitted at the University of Southampton titled: "\textit{Estimation and Inference in Seemingly Unrelated Systems of Nonstationary Quantile Predictive Regression Models}". The author is currently working on a follow-up paper titled: "\textit{Robust Identification and Estimation in Non-Gaussian Structural VARs with Near Unit Roots}". Financial support from the Research Council of Finland (grant 347986) is gratefully acknowledged. Address correspondence to Christis Katsouris, Faculty of Social Sciences, University of Southampton, United Kingdom. E-mail: \textcolor{blue}{\texttt{c.katsouris@soton.ac.uk}}   }  }
}

\author{\textbf{Christis Katsouris}\thanks{Dr. Christis Katsouris is a Postdoctoral Researcher at the Faculty of Social Sciences, University of Helsinki.   } \\  \textit{University of Southampton}\\ $\&$  \textit{University of Helsinki} \\  \\ \textcolor{blue}{Job Market Paper I}}

\date{\today}

\maketitle

\begin{abstract}
\vspace*{-0.25 em}
This article studies identification and estimation for the network vector autoregressive model with nonstationary regressors. In particular, network dependence is characterized by a nonstochastic adjacency matrix. The information set includes a stationary regressand and a node-specific vector of nonstationary regressors, both observed at the same equally spaced time frequencies. Our proposed econometric specification correponds to the NVAR model under time series nonstationarity which relies on the local-to-unity parametrization for capturing the unknown form of persistence of these node-specific regressors. Robust econometric estimation is achieved using an IVX-type estimator and the asymptotic theory analysis for the augmented vector of regressors is studied based on a double asymptotic regime where both the network size and the time dimension tend to infinity. 
\\

\\

\textbf{JEL} Classification: C12, C22
\end{abstract}

\newpage 

\setcounter{page}{1}
\pagenumbering{arabic}

\section{Introduction}

Network connectivity has important implications for the risk management of economic, financial and societal events such as the diffusion of spillover effects across large networks (e.g., \cite{he2018measuring}), the financial contagion in stock markets (e.g., \cite{hardle2016tenet}, \cite{chen2019tail}, \cite{mitchener2019network}), the detection of market exuberance (e.g., \cite{Magdal2009limit}) as well as the spread of epidemic diseases (e.g., \cite{keeling2005networks}). Time series network-driven models have seen growing attention in the literature as a statistical mechanism for identifying underline network structures based on methods commonly used for Vector Autoregression Models (e.g., \cite{zhu2017network, zhu2019network} and \cite{zhu2020grouped}).  Our study is about the aspects of identification and estimation of cointegration dynamics and the stability of autoregressive processes under network dependence. We focus on robust econometric estimation of the Network Vector Autoregression (NVAR), with nonstationary (near-unit-root) node-specific regressors.  

We consider a network with $N$ nodes (possibly high dimensional) such as a social or a financial network  which is indexed by $i \in \left\{ 1,...,N \right\}$. The network structure is characterized by a binary nonstochastic adjacency matrix $\boldsymbol{\Omega} = \left( \omega_{ij} \right) \in \mathbb{R}^{ N \times N}$ such that $\omega_{ij} = 1$ if a direct link between the pair $(i,j)$ exists and $\omega_{ij} = 0$ otherwise. We denote with $Y_{i(t)} \in \mathbb{R}$ the continuous response variable obtained from node $i$ at time point $t$ and respectively $\mathbb{Y}_t = \left( Y_{1t},..., Y_{Nt} \right)^{\top} \in \mathbb{R}^N$ the possibly high dimensional vector or responses for large $N$. Under the assumption of identifiable network dependence we assume that the response variable of the node $i$, $Y_{i(t)}$, is affected by its lag value $Y_{i,(t-1)}$, as well as by its connected nodes which are collected via $\left\{ j : \omega_{ij} = 1 \right\}$. Moreover, the notion of network-driven predictability is captured by a set of node-specific variables, denoted by $X_{i(t)} \in \mathbb{R}^p$ which represent regressors of abstract degree of persistence generated by local-to-unit root processes.

Our research objective is to study the relationship between the stationary node-specific regressands and a set of node-specific nonstationary regressors  that capture cointegration dynamics under network dependence. Our proposed framework aims to study economic phenomena under network dependence which implies that time series observations from the infinite past evolve conditioning on initial values. Thus we provide an analysis of discrete time, time invariant, and causal dynamic systems where time, following a finite sequence of initial values, is explicitly confined to positive integers. In other words, the unobserved random disturbances  determine how random influences enter the multivariate system. We can think of the random disturbances  of our system to be determining the shocks that enter the system under the simultaneous presence of cointegration dynamics and network dependence\footnote{Regardless of the presence of network dependence we don't study temporal network dynamics under nonstationarity but the simultaneous presence of nonstationarity and network dependence. Although, we don't extend our framework into a high dimensional setting, some relevant studies include \cite{adamek2022local}, \cite{chen2023community},  \cite{krampe2023structural} and \cite{zhang2023statistical} as well as \cite{barigozzi2023fnets}, \cite{cho2023high} and  \cite{fang2023determination}.} (see, also \cite{basu2023graphical}). Recently, \cite{bykhovskaya2022time} consider an estimation and prediction framework where the model specifies the temporal evolution of a weighted network that combines classical autoregression with non-negativity, a positive probability of vanishing and peer effect interactions between weights assigned to edges in the process.

\newpage

Although nonstationarity in time series can be interpreted as time-varying model parameters, we shall focus on cointegrating and unit root dynamics when the underline data structure is defined across nodes of a single network. We are interested for conditions of network stationarity against explosiveness of the underline network evolution process. Regarding suitable mixing condition, we shall conjecture whether the network effect dominate the nonstationary property of regressors or whether the network effects appear in the limiting distribution of the $\mathsf{NVAR}$ estimator.  

Generally, we aim to investigate aspects related to the relation between the network topology and the nonstationarity properties of regressors such as under which conditions the degree of centrality  as captured by the network-covariates and the nonstationarity as captured by the LUR parametrization asymptotically dominates when establishing the asymptotic theory of model parameters. In particular, we shall aim to derive related stability conditions of the system representation based on the network topology and the properties of the nonstationary time series regressors. In other words, some of the challenges that we need to address include the presence of network dependence and the long memory properties (e.g., see \cite{schennach2018long}) of the system under time series nonstationarity. Some specific questions of interest include the study of asymptotics and stochastic behaviour of processes in the case of  explosive regressors as well as the asymptotic theory analysis of the IVX filtration under the presence of network dependence. 

To examine both aspects in a unified framework, we begin by studying the limiting behaviour of the classical least squares estimator within the stationary framework under the assumption of network dependence. In particular, the presence of network effects (e.g., the centrality measure as a covariate), the inclusion of model intercepts (to capture mean nodal effects), the lagged response variable as well as the nonstationary regressors (such as persistent regressors), can induce different convergence rates which require to develop suitable asymptotic theory that includes these features. Firstly, we need to examine whether the network dependence affects the limiting behaviour of the OLS type estimator.  A starting point is to verify a Gaussian random variant as a limiting distribution for the stationary framework, in a similar fashion as in \cite{zhu2017network}). The stationary framework, allows to simplify the econometric specification by assuming that the nodes' regressors are stationary, so that the LUR process is not employed to model the regressors. 

Secondly, for the nonstationary framework, which is the case that the nodes' regressors are generated by the LUR specification we need to examine how the presence of network dependence affects the asymptotic theory of the OLS estimator. A non-Gaussian limiting distribution can indicate that the presence of both network effects and nonstationarity requires a different methodology to robustify inference. To this direction, we can further examine the implementation of the IVX instrumentation of \cite{PM2009econometric}, which has been proven to be robust under the assumption of temporal dependence, and provides a statistical methodology for filtering the unknown degree of persistence. Thus, by demonstrating that the asymptotic distribution of an IVX-type estimator under the assumption of both network dependence and nonstationarity, converges to a mixed Gaussian distribution, then this provides a robust methodology for estimation and inference for the NVAR model within the aforementioned framework.     


\newpage

\subsection{Contributions and Outline of the paper}

To the best of our knowledge our study is the first to consider a network-type of predictability in the same spirit as the conventional predictability literature based on  nonstationary time series regressions (see, \cite{kostakis2015Robust, kostakis2018taking}. The LUR parametrization is employed to capture the unknown form of persistence that these nonstationary regressors exhibit  and along with the network-dependent covariates, our functional form specification allows to model cointegration dynamics in such settings. Our proposed system representation is novel and thus the closer to our identification and estimation strategy is the framework presented in the study of \cite{magdalinos2021least}. In particular, within the predictive regression literature the regressand corresponds to a continuous stationary covariate such as stock returns while regressors correspond to persistent data\footnote{ A different stream of literature considers the case in which the response variable represents a time series of counts, which is useful when modelling corporate defaults as in \cite{agosto2016modeling}) and \cite{armillotta2022nonlinear} who consider modelling count data time series reponsens using nonlinear network vector autoregressions.}. Motivated by these observations we propose a network-dependent autoregressive representation under the presence of time-series nonstationarity while we consider suitable stability conditions for an increasing network size. We believe that our proposed system representation as well as identification and estimation strategy are of relevance both from the theoretical and applied econometrics perspective. 

Extending existing econometric specifications and estimation approaches to the case of $\mathsf{NVAR}(1)$ with nonstationary regressors  implies overcoming several challenges to ensure robust implementation approaches. First, it is not immediately clear which is the most suitable system representation that incorporates both nonstationary regressors and network dynamics, in the form of network covariates induced from an non-stochastic adjacency matrix.  Second, while asymptotic theory results for the case of nonstationary predictive regression models are well developed (see   \cite{jansson2006optimal}, \cite{phillips2013predictive}, \cite{breitung2015instrumental}, \cite{lee2016predictive}, \cite{kostakis2015Robust}, \cite{kasparis2015nonparametric}, \cite{andersen2021consistent} and \cite{liu2023robust} among others) and systems of predictive and cointegrated regressors (see \cite{Phillips2008limit}, \cite{magdalinos2009limit} and \cite{magdalinos2021least}), as well as limit results for the case of network vector autoregressive models under time series stationarity (see \cite{zhu2017network}) and inference techniques under cluster dependence in two or more dimensions (see \cite{menzel2021bootstrap}, \cite{olmo2023nonparametric}); there is no related asymptotic theory that combines these features in a unified framework (see, also discussion in \cite{katsouris2023limit, katsouris2023optimal}).  Third, and perhaps the most profound challenge that would need to be addressed is the fact that the presence of stationary and nonstationary regressors, in the form of partially nonstationary and cointegration dynamics (see, \cite{ahn1990estimation},  \cite{toda1995statistical},  \cite{cavanagh1995inference}, \cite{paruolo1997asymptotic} and \cite{poskitt2006identification}) along with near unit roots (see, \cite{phillips2007limit}), satisfy certain properties and regularity conditions which would otherwise be violated under the presence of network dependence, if these combined features are not being taken care properly\footnote{Suitable markovian conditions are commonly used to preserve the time series structure when network dependence is jointly modelled in the time series dimension. Although the markovian property, as captured by the transition matrix. can be expressed in terms of spatial dependence, in this paper we consider a more general form of dependence without restricting the possible dependence structures only to spatio-temporal processes.} (see, also \cite{white2000asymptotic}).

\newpage

The aforementioned challenges have an impact on inferential procedures and especially due to the well-known problems with obtaining uniform inference results to the conventional literature of predictive and cointegrating regression models (see,     \cite{mikusheva2007uniform} and \cite{phillips2014confidence}). Thus, our contributions in this paper are summarized as below:
\begin{itemize}

\item Our econometric environment corresponds to multivariate time series processes under network dependence. We develop asymptotic theory and estimation techniques for inference under two large-sample regimes such that $(i)$ with increasing time sample size, $T \to \infty$, and fixed network dimension, denoted by $N$, and $(ii)$ with $N \to \infty$ and $T_N \to \infty$, where the temporal size depends on $N$. In other words, the case in which both indices tend to infinity has certain challenges when deriving the asymptotic properties or such complex tail dependent processes. 

\item Our asymptotic theory analysis is developed based on the classical joint weak convergence arguments of functionals of Brownian motion equipped with $J_1$ topology of $\mathcal{D}_{ \mathbb{R}^p } \left( [0,1] \right)$, which holds regardless of the presence of network dependence. Moreover, we carefully explain the additional necessary conditions to ensure such that such topological convergence results are still valid in a more general autoregressive environment.  

\end{itemize}
The organization of the paper is as follows. Section 2 introduces the Network Vector Autoregressive Model with nonstationary Regressors. Section 3 develops large-sample theory for the proposed modelling environment. Section 4 explains the Monte Carlo Simulation study of the paper and Section 5 explains the empirical application. Section 6 concludes. Technical proofs can be in the \hyperref[appA]{Appendix} of the paper.
For any real arbitrary matrix $\boldsymbol{A}$, the norm is denoted by $\norm{ \boldsymbol{A} }$ and corresponds to the Frobenius norm defined by $\norm{ \boldsymbol{A} } = \sqrt{ \mathsf{trace} ( \boldsymbol{A}^{\prime} \boldsymbol{A} ) }$. Let $\lambda_i ( \boldsymbol{M} )$ denote the $i-$th largest eigenvalue of an $( n \times n)$ symmetric matrix $\boldsymbol{M}$ with its eigenvalues such that $\lambda_1 ( \boldsymbol{M} ) \geq ... \geq  \lambda_n ( \boldsymbol{M} ) $. The spectral norm of $\boldsymbol{A} $ is denoted by $\norm{ \boldsymbol{A}  }_2$, such that, $\norm{ \boldsymbol{A}  }_2 = \sqrt{ \lambda_1 ( \boldsymbol{A}^{\prime} \boldsymbol{A}  )  }$, is maximum column sum norm is denoted by $\norm{ \boldsymbol{A} }_1$, such that, $\norm{ \boldsymbol{A} }_1 = \mathsf{max}_{ 1 \leq j \leq n}  \sum_{i = 1}^m | A_{ij} |$ and its maximum row sum norm is denoted by $\norm{ \boldsymbol{A} }_{ \infty }$, such that, $\norm{ \boldsymbol{A} }_{ \infty } = \mathsf{max}_{ 1 \leq i \leq n}  \sum_{i = 1}^m | A_{ij} |$.  Moreover, the operator $\overset{P}{\to}$ denotes convergence in probability, and $\overset{D}{\to}$ denotes convergence in distribution.

\subsection{Prior and related work}

A large stream of literature methodologies for modelling cross sectional dependence and heterogeneity via the use of dynamic panel models. In particular,  \cite{kapetanios2014nonlinear} develop asymptotic theory for  nonlinear panel models with cross-sectional dependence. Moreover, \cite{huang2020two} propose a network autoregressive model for two-mode networks where the econometric specification allows for different network autocorrelation coefficients such that 
\begin{align}
Y_1 &= \rho_{12} W_{12} Y_2 + X_1 \beta_1 + \varepsilon_{1},
\\
Y_2 &= \rho_{21} W_{21} Y_1 + X_2 \beta_2 + \varepsilon_{2}, 
\end{align}

\newpage

where $\varepsilon_{1}$ and $\varepsilon_{2}$ are assumed to be independent. In other words, the particular structure measures the response collected for the $k-$th group of nodes. Then, it can be shown that an equivalent model representation is given by $Y = \left(  I_N - \mathbb{W}_{\rho} \right)^{-1} ( \mathbb{X} \beta + \varepsilon )$, where $N = n_1 + n_2$ corresponds to the network size as the sum of the sample size across the two groups. The estimation of the model parameters is achieved using the quasi-maximum likelihood approach although robustness is achieved when heteroscedasticity is properly modelled (see, also  \cite{anufriev2015connecting}).  On the other hand, \cite{huang2020two} propose an approximation obtained via the least squares estimator to handle the computational complexity due to the large-scale network while in our case the IVX estimator is known to handle the aspects of endogeneity and unknown form of persistence. 

Further examples include the multivariate spatial autoregressive model for large scale social networks proposed by \cite{zhu2020multivariate}.  All aforementioned studies correspond to modelling the relationship between the node-specific responses and a set of exogenous covariates for each node. Recently, \cite{meitz2022subgeometrically} propose subgeometrically ergodic autoregressions with autoregressive conditional heteroscedasticity. The particular time series approach is applicable for modelling univariate nonlinear autoregressions with autoregressive conditional heteroscedasticity of the form: $y_t = \alpha_1 y_{t-1} + ... + \alpha_p y_{t-p} + g( u_{t-1}) + \sigma_t \varepsilon_t$. The types of models we consider in this paper are of the form: $y^{(i)}_t =  \beta_0 + \beta_1 n_i^{-1} \sum_{j=1}^N \alpha_{ij} y^{(j)}_{t-1} + \beta_2 y_{t-1}^{(i)} + \varepsilon_t^{(i)}$  as well as $y^{(i)}_t =  \beta_0 + \beta_1 n_i^{-1} \sum_{j=1}^N \alpha_{ij} y^{(j)}_{t-1} + \beta_2 y_{t-1}^{(i)} + X^{\prime}_i \gamma +  \varepsilon_t^{(i)}$, with the only difference that any node-specific covariates, as  specified by \cite{zhu2017network} is replaced by a set of nonstationary regressors. Another relevant framework is proposed by \cite{nicholson2017varx} who consider a high-dimensional VAR model with exogenous variables with a vector of parameters estimated based on the objective function $\underset{ \boldsymbol{\mu}, \boldsymbol{\Phi}, \boldsymbol{\beta}  }{ \mathsf{arg \ min} } \sum_{t=1}^T \norm{ \boldsymbol{y}_t - \boldsymbol{\mu} - \sum_{i=1}^p \boldsymbol{\Phi}_{(i)} \boldsymbol{y}_{t-j} - \sum_{j=1}^q \boldsymbol{\beta}_{(j)} \boldsymbol{x}_{t-j} }$. Moreover, the authors consider an extension to unit-root dynamics but they tackle the presence of nonstationarity using stationarity transformations which have the disadvantage of destroying information about the long-run dynamics of regressors with possible cointegration and nonstationarity. 

In our study we study the presence of nonstationarity via the local-to-unity parametrization to model persistent data which we convert to mildly integrated using the IVX filtration which controls the degree of endogeneity between the innovation term of predictors and the innovation term of the predictive regression.  Our proposed approach aims to unify these two aspects; (i) the network structure, characterized by the adjacency matrix, which allows network dependence in Vector Autoregression models, and (ii) the persistence properties of regressors pioneered with the seminal work of \cite{Phillips1987time, Phillips1987towards} and \cite{phillips2007limit}. To the best of our knowledge, the inclusion of persistence properties of regressors for the development of the asymptotic theory of the NVAR model its a novel contribution to the literature regardless of employing limit theory and stochastic calculus algebra from existing studies for the development of our asymptotic theory. The studies of \cite{katsouris2021optimal} and \cite{katsouris2023estimating, katsouris2023statistical}, motivated from the aspects of financial connectedness (see, \cite{diebold2014network} and \cite{barunik2018measuring}) and systemic risk were the first to bridge the gap in these two streams of literature by modelling the persistence when estimating systemic risk measures as the tail risk measures in \cite{Adrian2016covar} and \cite{hardle2016tenet}.

\newpage

\section{Econometric Identification and  Estimation }

\subsection{Network Vector Autoregression}
\label{section2}

Consider a high dimensional network of size $N$ and $Y_{i(t)}$ to be the response variable which corresponds to node $i$ obtained at time $t$. For each node $i$, we assume the existence of a $p-$dimensional node specific random vector of regressors $X_{i(t)} = \left( X_{i1}, ...,  X_{ip} \right)^{\prime} \in \mathbb{R}^p$. Therefore, to model $Y_{i(t)}$, we propose the following $\mathsf{NVAR}(1)$ model 
\begin{align}
\label{model1}
y_{i(t)} &= \beta_0 + \underbrace{ \beta_1 n_i^{-1} \sum_{j=1}^N \omega_{ij} y_{j(t-1)} + \beta_2 y_{i(t-1)} }_{(\text{network effect})}  + A X_{i(t-1)} + \epsilon_{i(t)}
\\
\label{model2}
\boldsymbol{X}_{i(t)} &= \boldsymbol{R}_n \boldsymbol{X}_{i(t-1)}  + \boldsymbol{U}_{i(t)}
\end{align}
with $1 \leq t \leq T$, where the autoregressive coefficient matrix $\boldsymbol{R}_n$ is expressed as below
\begin{align}
\label{model3}
\boldsymbol{R}_n &= \left( \boldsymbol{I}_p - \frac{\boldsymbol{C}_p }{ T} \right),
\end{align} 
with $\boldsymbol{C}_p = \text{diag} \left\{ c_1,..., c_p \right\}$ such that $c_i$ denotes the unknown coefficient of persistence such that $n_i = \sum_{ j \neq i} \omega_{ij}$ is the total number of nodes that $i$ follows, which is the out-degree measure. 
\begin{align}
\begin{bmatrix}
Y_{1}^{i}  \\
\vdots    \\
\vdots     \\
Y_{T}^{i} 
\end{bmatrix}
=
\begin{bmatrix}
1   & 1   &    & ... & 0   \\
... & ... &    & ... & ... \\
... & ... &    & ... & ... \\
1   & 1   &    & ... & 0 
\end{bmatrix}
\begin{bmatrix}
X_{1(t-1)}^{i}  \\
\vdots    \\
\vdots     \\
X_{p(t-1)}^{i} 
\end{bmatrix}^{\prime} +
\begin{bmatrix}
U_{1(t)}^{i}  \\
\vdots    \\
\vdots     \\
U_{p(t)}^{i} 
\end{bmatrix}
\end{align}
In matrix form we have that the regressors are generated via 
\begin{align}
\begin{bmatrix}
X_{1(t)}^{i}  \\
\vdots    \\
\vdots     \\
X_{p(t)}^{i} 
\end{bmatrix}^{\prime}
=
\begin{bmatrix}
\left( 1 - \frac{c_1}{ T} \right) & 0 & ... & ... & 0 \\
0 & \left( 1 - \frac{c_2}{ T } \right) & 0 & ... & 0 \\
0 & ... & .... & ... & 0 \\
0 & ... & .... & ... &  \left( 1 - \frac{c_p}{ T } \right)
\end{bmatrix}
\begin{bmatrix}
X_{1(t-1)}^{i}  \\
\vdots    \\
\vdots     \\
X_{p(t-1)}^{i} 
\end{bmatrix}^{\prime} +
\begin{bmatrix}
U_{1(t)}^{i}  \\
\vdots    \\
\vdots     \\
U_{p(t)}^{i} 
\end{bmatrix}
\end{align}
for $i \in \left\{ 1,..., N \right\}$. The NVAR(1) model given by \eqref{model1}-\eqref{model3} implicitly assumes that a particular node $i$ can be affected by another node $j$, if and only if the pair $(i,j)$ are interconnected as described by the nonstochastic adjacency matrix $ \left[ \boldsymbol{\Omega} \right]_{ij}$. Therefore, in this paper we propose a network vector autoregression (NVAR) model with nonstationary regressors. The $\mathsf{NVAR}(1)$ model assumes that each node's response at a given time time point is a linear combination of (a) its lag value, (b) the average of its connected neighbours, (c) a set of node-specific regressors; and (d) an independent noise. 

\newpage

\subsection{Statistical Framework}

Recall that $N$ is the network size and $Y_{it}$ is the stationary regressand that corresponds to the $i-$th subject at time point $t$. In particular, we focus on the statistical estimation of $\mathsf{NVAR}$ models with nonstationary regressors for the case that their asymptotic properties depend on both $N$ (size) and $T$ (time). Therefore, when the dependence of asymptotic approximations to sample moments relies on both indices, such that $N$ and $T$, tend to infinity, it requires us to employ a double asymptotic regime. At the same time an additional challenge we need to address is the presence of nonstationary regressors when deriving asymptotic theory. However, within our proposed econometric framework and in contrast to the study of \cite{zhu2017network} and \cite{zhu2020grouped}, we replace the $p-$dimensional node-specific random vector with a $p-$dimensional node-specific time indexed vector of nonstationary regressors, which is assumed to be generated using a local-to-unity parametrization. Specifically, the LUR parametrization for the nonstationary regressors of the model along with the network-dependent covariates aims to model cointegration dynamics.  

To the best of our knowledge our study is the first to consider a network-type of predictability in the same spirit as the conventional predictability literature based on  nonstationary time series regressions. In particular, within the predictive regression literature the regressand corresponds to a continuous stationary covariate such as stock returns while regressors correspond to persistent data\footnote{ A different stream of literature considers the case in which the response variable represents a time series of counts, which is useful when modeling for example corporate defaults as in \cite{agosto2016modeling}) as well as \cite{armillotta2022nonlinear} who consider modelling count data time series reponsens using nonlinear NVARs.}. Motivated from both of these observations we propose a network-dependent representation of network interactions under the presence of time-series nonstationarity while we consider suitable stability conditions for increasing network size. We believe that our proposed system representation as well as identification and estimation strategy are of relevance both from the theoretical and applied econometrics perspective. In terms of asymptotic theory properties we consider sufficient conditions for the information matrix to be stochastic equicontinuous. Moreover,  the main challenge is that the presence of the nuisance parameter of persistence cannot be consistently estimated which is an aspect discussed in several papers found in the predictive regression literature.  In particular, for a sequence of parameters $\boldsymbol{R}(n) = \boldsymbol{I}_p - \boldsymbol{C}_p / n$, where the real part of the eigenvalues of $\boldsymbol{C}_p  \in \mathbb{R}^{ p \times p }$ are all strictly negative, then the statistical problem is equivalent to estimating $\boldsymbol{C}_p = n \big( \boldsymbol{I}_p - \boldsymbol{R}(n) \big)$ which implies that in this setting the matrix $ \boldsymbol{R}(n) $ can only be estimated at a rate of $O ( n^{-1} )$. 

Within our framework we consider the uniform convergence of random variables. Then, for a $\mathsf{VAR}(1)$ process of the form $\boldsymbol{X}_t = \boldsymbol{R}(n) \boldsymbol{X}_{t-1} + \boldsymbol{\varepsilon}_t$, we denote with $\boldsymbol{\vartheta} := ( \boldsymbol{R}, \boldsymbol{\Sigma} )$ where $\boldsymbol{\Sigma} := \mathbb{E} \left(  \boldsymbol{\varepsilon}_t \boldsymbol{\varepsilon}_t^{\prime} | \mathcal{F}_{t-1} \right)$.  Notice that for for every set $\boldsymbol{\vartheta}$, there are infinitely many processes $\boldsymbol{X}_t$ and $\boldsymbol{\varepsilon}_t$ satisfying the assumptions above. Consequently, by imposing certain requirements on the parameter $\boldsymbol{\vartheta}$ and specifically on the real-valued matrix $\mathbb{R}$, then all these processes are ensured to exhibit the same limiting behaviour.

\newpage 

We shall clarify that this is not a parameter restriction (as commonly done in the SVAR literature), but simply an a priori assumption regarding the asymptotic behaviour of certain stochastic processes (e.g., see ).  Moreover, for every $\boldsymbol{\vartheta} \in \Theta$, suppose that there exists matrices $\boldsymbol{\Gamma} \in \mathbb{C}^{ p \times p  }$ and $\boldsymbol{J} \in \mathbb{C}^{ p \times p  }$ such that $\boldsymbol{J} $ is a Jacobian matrix and it holds that $\boldsymbol{R}(n) \equiv \boldsymbol{\Gamma} \boldsymbol{J} \boldsymbol{\Gamma}^{-1}$. Then, up to the reordering of the eigenvalues, the matrix $\boldsymbol{J} $ is unique and satisfies $\boldsymbol{J}  \in \mathcal{J} \big(  |  \lambda_1   | \big)$ which is called the Jordan canonical form of the LUR matrix $\boldsymbol{R}(n)$ (see, \cite{dou2021generalized}).   

Consider the LUR parametrization of a $\mathsf{VAR}(1)$ model expressed as below
\begin{align*}
\boldsymbol{X}_t = \left(  \boldsymbol{I}_p - \frac{ \boldsymbol{C}  }{ n } \right) \boldsymbol{X}_{t-1} + \boldsymbol{u}_t
\end{align*}
Notice that we denote with $\boldsymbol{\Xi}_{xx} := \mathbb{E} \big[  \boldsymbol{X}_t  \boldsymbol{X}_t   ^{\prime}  |  \mathcal{F}_{t-1}  \big]$, which satisfies 
\begin{align*}
\boldsymbol{\Xi}_{xx} = \xi^2 \left(  \boldsymbol{I}_{p^2} - \left(  \boldsymbol{I}_p - \frac{ \boldsymbol{C}  }{ n } \right) \otimes \left(  \boldsymbol{I}_p - \frac{ \boldsymbol{C}  }{ n } \right)  \right)^{-1}.
\end{align*}
Moreover, we can express the diagonal matrix that includes the nuisance parameters of persistence 
\begin{align*}
\boldsymbol{C} &\equiv \boldsymbol{\Gamma} \boldsymbol{J} \boldsymbol{\Gamma}^{-1}
\\
\left(  \boldsymbol{I}_p - \frac{ \boldsymbol{C}  }{ n } \right)  &\equiv  \boldsymbol{\Gamma}   \left(  \boldsymbol{I}_p - \frac{ \boldsymbol{J}_p  }{ n  }  \right) \boldsymbol{\Gamma}^{-1}
\end{align*}

In a similar fashion, lets suppose that we have the local-to-unity model of a univariate time series such that $x_t = \left(  1 - \frac{c}{n} \right) x_{t-1} + u_t$, then if we assume that there exists $p$ autoregressive roots such that 
\begin{align*}
\left(  1 - \rho_1 L \right) \left(  1 - \rho_2 L \right) ... \left(  1 - \rho_p L \right)  \big( x_t - \mu \big) \equiv u_t.
\end{align*}
where there is a sequence of nuisance parameters of persistence $\left\{ c_j \right\}_{j=1}^p$. Notice that the above expression corresponds to the case of a univariate time series which is expressed in the form of $p$ autoregressive near unit roots in the similar spirit as an $AR(p)$ model has $p$ stationary roots. On the other hand, a VAR$(1)-$model corresponds to the case where we have, say, a $d-$dimensional (multivariate) series $\boldsymbol{x}_t$, which has near unit roots. Similarly, a $\mathsf{VAR}(p)-$model implies the presence of $p-$LUR across the $d-$dimensional vector, which is again different than a multivariate $\mathsf{VAR}(p)-$model with stationary roots (i.e., meaning stable without any LUR parametrization). There are of course the econometric specifications of a $\mathsf{SVAR}(1)$, $\mathsf{SVAR}(p)$ under time series stationarity and $\mathsf{SVAR}(p)$ under time series nonstationarity which are all representations beyond the scope of this paper. Next, we shall focus on the required assumptions that we need to impose on the eigenvalues of the sample matrices.  Notice that we assume that uniform convergence of random matrices means uniform convergence of the vectorization of these matrices.  Although it is not trivial to conduct inference on the matrix $\boldsymbol{R}_n$ due to the presence of the nuisance parameter $C_n$, which cannot be uniformly consistently estimated, the IVX estimator of the matrix induces mixed Gaussian asymptotics.

\newpage

\subsection{Main Assumptions}

\begin{assumption}
Within the stationary framework we impose the following assumptions: 
\begin{itemize}
\item[\textbf{(C1)}] \textbf{(Nodal Assumption)} Assume that $Z_i$'s are \textit{i.i.d} random variables, with mean 0 and covariance $\Sigma_z \in \mathbb{R}^{p \times p}$ and finite fourth order moments. Moreover, for this case $\left\{ Z_i \right\}$ and $\left\{ \epsilon_{it} \right\}$ are assumed to be mutually independent.

\item[\textbf{(C2)}] \textbf{(Network Structure)} Assume that $\left\{ \boldsymbol{\Omega}_i \right\}$ is a sequence of nonstochastic matrices such that $i \in \left\{ 1,...,N  \right\}$.

\begin{itemize}
\item[\textbf{(C2.1)}] \textbf{(Connectivity)} Treat $\boldsymbol{\Omega}$ as a transition probability matrix of a Markov Chain, whose state space is defined as the set of all the nodes in the network, for $i \in \left\{ 1,...,N  \right\}$. We assume that the Markov chain is irreducible and aperiodic. Furthermore, we define $\pi = \left( \pi_1,..., \pi_N \right)^{\prime} \in \mathbb{R}^N$ as the stationary distribution of the Markov chain, such that (a) $\pi \geq 0$ and $\sum_{i = 1}^N \pi_i = 1$, (b) $\pi = W^{\prime} \pi$. Moreover, $\sum_{i=1}^N \pi_i^2$ is assumed to converge to 0 as $N \to \infty$.  

\item[\textbf{(C2.2)}] \textbf{(Uniformity)} Define $W^{*} = W + W^{*}$ as a symmetric matrix. Assume that $\lambda_{\text{max}} \left( W^{*} \right) = \mathcal{O} \left( \text{log} N \right)$, where $\lambda_{\text{max}} \left( A \right)$ denotes the largest absolute eigenvalue of an arbitrary symmetric matrix $M$.  
\end{itemize}

\item[\textbf{(C3)}] \textbf{(Law of Large Numbers)} Define $Q := ( I - G )^{-1} ( I - G ^{\prime})^{-1}$, and recall $G = \beta_1 W + \beta_2 I$. Then, assume that the following limit exists: $\kappa_1 = \text{lim}_{ N \to \infty} N^{-1} \text{trace} \left\{ \Gamma(0) \right\}$, $\kappa_2 = \text{lim}_{ N \to \infty} N^{-1} \text{trace} \left\{ W \Gamma(0) \right\}$, $\kappa_3 = \text{lim}_{ N \to \infty} N^{-1} \text{trace} \left\{ (I - G)^{-1} \right\}$ and $\kappa_4 = \text{lim}_{ N \to \infty} N^{-1} \text{trace} \left\{ Q \right\}$, where $\kappa_1, \kappa_2, \kappa_3$ and $\kappa_4$ are fixed constants.
\end{itemize}
\end{assumption}

\begin{remark}
Some relevant remarks regarding our conditions are summarized below: 
\begin{itemize}

\item Condition (C1) provides a basic assumption regarding the nodal regressors and the innovation sequence $\epsilon_{it}$. Moreover, we can check how to correctly impose the covariance structure and the estimation of the corresponding covariance matrices as in PM. The CLT above does not apply in the case of stochastic integrals (only valid for the stationary framework). 

\item  Condition (C2) is related to the network structure, as characterized by the adjacency matrix $W$. Firstly, condition (C2.1) assumes that all nodes are reachable to each other, that is, irreducibility property. More specifically, for two arbitrary nodes  $i$ and $j$ a path of finite length connecting $i$ and $j$ should exist. A simple sufficient condition for both irreducibility and aperiodicity is that the network is always fully connected after a finite number of steps. That is, there exists an $n^{*}$ such that , for any $n \geq n^{*}$, each component in $W^n$ is always positive. Secondly, condition (C2.2) requires that the network structure should admit certain uniformity property so that the diverging rate $\lambda_{ \text{max}} \left( W^{*} \right)$ should be sufficiently slow. 

\end{itemize}
\end{remark}

\newpage

For stationary and ergodic time series models the common practice is to establish the existence and uniqueness of solutions of related recurrent equations of the form $\boldsymbol{X}_t = \boldsymbol{A}_t \boldsymbol{X}_{t-1} + \boldsymbol{B}_t, t \in \mathbb{Z}$, as in \cite{fort2005subgeometric}, \cite{meitz2021subgeometric, meitz2022subgeometrically} and \cite{matsui2022characterization} who are exploiting the property that a class of BEKK-ARCH processes have multivariate stochastic recurrence equation representations to show the existence of strictly stationary solutions under mild conditions (see, also \cite{meitz2008ergodicity, meitz2008stability}). Moreover, \cite{doukhan2023stationarity} study the stationarity and ergodic properties for observation-driven models in random environments using strict exogeneity assumptions. Specifically, the existence of stationary solutions in the sequential exogeneity framework often relies on additional Lipschitz type properties. However, in our study we consider a specific functional form which includes both nonstationary regressors as well as network-dependent covariates, thereby making it more challenging to either establish a direct recurrent equation representation or employing other VAR representations as done in the SVAR literature.

\subsubsection{Stability Conditions}

An alternative way for classifying the stability conditions of the autoregressive processes which is commonly done by the assumption on the persistence properties of regressors as in \cite{kostakis2015Robust} and \cite{magdalinos2021least}, is to consider relevant eigenvalue conditions on the LUR coefficient matrix which are based on whether the ordered eigenvalues are below unity or not (e.g., see \cite{holberg2023uniform}). Nevertheless, such conditions can be also thought as a direct implication of the persistence class given by \cite{PM2009econometric} conditions and the spectral radius conditions in M $\&$ PCB (2020). 
 
Consider the case where we split the parameter space $\Theta$ into two overlapping regions with respect to the eigenvalues such that 
\begin{align*}
R_n^{(1)} := \left\{ \theta \in \Theta  : | \lambda_1  | \leq 1 - \frac{ \mathsf{log}  }{ n }  \right\}, \ \ \  R_n^{(0)} := \left\{ \theta \in \Theta  : | \lambda_1  | \leq 1 - \frac{ \mathsf{log}  }{ n }  \right\}
\end{align*}

Notice that for multivariate time series settings cointegrated systems or almost cointegrated systems in the case where the roots are only close to unity. Therefore, due to the presence of asymptotic distribution discontinuities in order to ensure limit results of known form proper normalizations of the sample covariances are required\footnote{In particular, in Structural Vector Autoregressive models with nonstationary regressors in order to ensure valid uniformly inference over a set of parameters including processes that are cointegrated and with roots arbitrarily close to the unit circle, adaptive estimation and inference procedures are needed, although that case is beyond the scope of this paper which are  currently investigated in a different paper.}. Usually in the literature the assumption regarding the LUR matrix is that an assumption is imposed that $\boldsymbol{R}_n$ is either a drifting sequence of diagonal matrices or a fixed diagonal matrix. We focus on proving that the asymptotic distributions of the relevant sample covariances can be approximated by stochastic integrals of Orneisten-Uhlenbeck processes.  Therefore, we aim to better understand what is the relation between  network dependence and the unknown degree of persistence.

\newpage

In particular, it is the case that the OLS estimator is not standard normal distribution asymptotically, then we can examine the implementation of the IVX estimator as a filtration method for both the persistence and the network dependence.  In other words, we consider robust estimation and inference methods for a Network Vector Autoregression Model with LUR regressors. The model we consider is closely related to the NAR model proposed by \cite{zhu2017network} and \cite{zhu2019network} as well as the frameworks proposed by \cite{magdalinos2009limit}.

\subsection{Limit theory for Vector Autoregression} 

\subsubsection{Weak Convergence Results}

To derive the asymptotic theory we employ the local-unit-root specification proposed by \cite{phillips1987time}. The regressors are assumed to be generated via the LUR process $X_{t} = \left( I_p - \frac{C_p}{T^{\lambda} } \right)$ $X_{t-1} + U_t$, with $\lambda < 1, \lambda \in (0,1) \ \text{or} \ \lambda > 1$. For example, for the case $\lambda = 1$, we consider the following $p-$dimensional Gaussian process
\begin{align}
\displaystyle J_C(r) = \int_0^r e^{C(r-s)} d B_u(s), \ \ \ \ r \in (0,1).
\end{align} 
which satisfies the Black-Scholes differential equation
$d J_c(r) \equiv C J_C(r) + d B_u(r)$, with $J_C(r) =0$, implying also that $J_C(r) \equiv \sigma_v J_C(r)$, where $\displaystyle J_C(r) =  \int_0^r e^{C(r-s)} d W_u(s)$ and $J_C(r)$ the Ornstein-Uhlenbeck, (OU) process,\footnote{The OU is a stationary Gaussian process with an autocorrelation function that decays exponentially over time. Moreover, the continuous time OU diffusion process has a unique solution.}
which encompasses the unit root case such that $J_C(r) \equiv B_u(r)$, for $C = 0$. 
 
\begin{theorem}[Uniform Convergence of Covariance Matrices (see, \cite{holberg2023uniform})]
\

Under assumptions above, and for the enlarged probability space $( \Omega, \mathcal{F}, \mathbb{P}  )$, there exists a standard $d-$dimensional Brownian motion, denoted by $\left\{ \boldsymbol{W}(t) \right\}_{ t \in [0,1]  }$, and a family of stochastic processes $\left\{ \boldsymbol{J}_{\boldsymbol{C}} (t) \right\}_{ t \in [0,1], n \in \mathbb{N} }$ such that 
\begin{align*}
\boldsymbol{J}_{ \boldsymbol{C}  } (t) = \int_0^t e^{ (t-s) \boldsymbol{C}  } \boldsymbol{\Gamma} \boldsymbol{\Sigma}^{1/2} d  \boldsymbol{W}(s), \ \ \text{with} \ \  \boldsymbol{J}_{ \boldsymbol{C} }(0) = \boldsymbol{0}.
\end{align*}
\end{theorem}
However, for the remaining of this paper, we consider the special case in which $\boldsymbol{\Gamma}_p \equiv \boldsymbol{I}_p$, is the identify matrix.  Moreover, suppose that the LUR matrix $\boldsymbol{R}(n)$ can be decomposed as   $\boldsymbol{R}(n) = \begin{bmatrix}
\lambda & 1
\\
0 & \lambda^{\prime}
 \end{bmatrix}$. As we have previously mentioned, for every $\boldsymbol{X}_t$ generated by $\boldsymbol{R}(n) \in \mathbb{R}^{ p \times p }$, there exists $\boldsymbol{\Gamma} \in \mathbb{C}^{ p \times p }$ and $\boldsymbol{J} \in \mathcal{J}$ such that $\boldsymbol{\Gamma}$ is invertible with $\boldsymbol{\Gamma}^{-1}  \boldsymbol{J} \boldsymbol{\Gamma}^{-1} = \boldsymbol{R}(n)$. We begin our asymptotic theory analysis by considering sequence of parameters in the stationary region of the LUR matrix $\boldsymbol{R}(n)$.

\newpage 
 
\begin{theorem}
Suppose that $\boldsymbol{V} \sim \mathcal{N} \left( \boldsymbol{0}, \boldsymbol{I}_{ p^2 } \right)$. For all $\delta > 0$ and $r \in [0,1]$, it holds that 
\begin{align*}
\underset{ n \to \infty  }{ \mathsf{lim}  } \ \underset{ \theta \in \Theta }{ \mathsf{sup} } \ \mathbb{P} \left( \norm{ \frac{1}{n} \boldsymbol{M}^{-1/2} \left(  \sum_{t=1}^{ \floor{nr} } \boldsymbol{X}_{t-1} \boldsymbol{X}_{t-1}^{\prime}  \right) \boldsymbol{M}^{-1/2} - r \boldsymbol{I}_p    } > \delta \right) = 0. 
\end{align*}
In other words, for the special case that $r = 1$ (full sample sum), then the above equation shows that the matrix $\boldsymbol{S}_{xx}$ converges in probability to the identity matrix. 
\end{theorem} 
Some further useful results include the following 
\begin{align*}
\underset{ t \geq 1 }{ \mathsf{sup}  } \ \underset{ \theta \in \Theta }{ \mathsf{sup} } \ \norm{ \mathbb{E} \big[  ( \boldsymbol{I} - \boldsymbol{R}_n )^{1/2}  \boldsymbol{X}_t \boldsymbol{X}_t^{\prime} ( \boldsymbol{I} - \boldsymbol{R}_n )^{1/2} \big]  } < \infty.
\end{align*} 
Furthermore, it holds that 
\begin{align*}
\underset{ t \geq 1 }{ \mathsf{sup}  } \ \underset{ \theta \in R }{ \mathsf{sup} } \ \norm{  ( \boldsymbol{I} - \boldsymbol{R}_n  )^{1/2}  \boldsymbol{M}  ( \boldsymbol{I} - \boldsymbol{R}_n  )^{1/2}  - \boldsymbol{\Sigma}_x } = 0. 
\end{align*}
where $\mathsf{vec} ( \boldsymbol{\Sigma}_x  ) = ( \boldsymbol{I} - \boldsymbol{R}_n  )^{1/2} \otimes ( \boldsymbol{I} - \boldsymbol{R}_n  )^{1/2 \prime} \big( \boldsymbol{I} - \boldsymbol{\Gamma} \otimes \boldsymbol{\Gamma}^{\prime} \big)^{-1} \mathsf{vec} (\boldsymbol{\Sigma})$.

\newpage

\subsubsection{IVX instrumentation}

\begin{example}
Suppose that $\boldsymbol{X}_t$ is a $\mathsf{VAR}(1)-$process, then we can split $X_t = \left(  Y_t, \tilde{X}_t  \right)^{\prime}$  (see also \cite{magdalinos2021least}) and the error term $u_t = ( u_{yt}, u_{xt} )^{\prime}$ into the first coordinate and their last $(p-1)$ coordinates  
\begin{align*}
\boldsymbol{Y}_t &= \boldsymbol{\beta} \tilde{\boldsymbol{X}}_{t-1} + u_{yt}, \tilde{\boldsymbol{X}}_t := \boldsymbol{\Gamma}  \boldsymbol{X}_t
\\
\tilde{\boldsymbol{X}}_t &= \boldsymbol{R}_n \tilde{\boldsymbol{X}}_{t-1} + \boldsymbol{u}_{xt} 
\end{align*} 
\end{example}

Furthermore, we consider the IVX instrumentation procedure which requires that 
\begin{align}
\widetilde{ \boldsymbol{Z} }_{i(t)} = \sum_{j=0}^{t-1} \left(  \boldsymbol{I}_p + \frac{ \boldsymbol{C}_z }{T^{\lambda}}  \right) \big( \boldsymbol{X}_{i(t-j)} - \boldsymbol{X}_{i(t-j-1)} \big)
\end{align}
for some $\boldsymbol{C}_z = \text{diag} \left\{ c_{z1},..., c_{zp} \right\}$ with $c_{zi} < 0$ and $0 < \kappa < 1$.  Notice that the particular instrumentation method is considered to be a linear filtering transformation of the regressor $X_{i(t)}$ into mildly integrated process, as it is explained in \cite{phillips2007limit} and  \cite{PM2009econometric}. However, the novelty here is the fact that we apply the IVX instrument across the nodes of the network, which ensures that the regressos of the NVAR(1) model have less degree of persistence than the original ones. Therefore, the development of a suitable \textit{FCLT} requires to consider the weakly convergence of the matrix moments which are based on the Network Vector Autoregression.  

Furthermore, due to the presence of the nuisance parameter $c_i$ which appear in the LUR process the OLS estimator $\widehat{\theta}$ will be biased (i.e., second degree bias), especially under the assumption of nonzero covariance terms in the covariance matrix of the innovations of the system. By substituting the IVX instrument, estimated in the second estimation step before obtaining the OLS counterpart in the first step, we allow for the abstract degree of persistence to be filtered out. Moreover, for simplicity in the derivations of the asymptotic theory we assume that the regressors for all nodes in the network are identical and belong to the same persistence class as defined by PM.In particular, the IVX instrumentation corresponds to endogeneously generated instruments to slow down the rate of convergence of the estimator enough to ensure mixed Gaussian limiting distributions, which is based on the assumption that all roots converge to unity at the same rate. This significant restriction since it implies that it excludes cases where parts of the process are stationary and other exhibit unit root behaviour such as mixed integration order, but nevertheless there are some studies which consider regressors of mixed integration order (see, \cite{phillips2013predictive} and M $\&$ PCB (2022)).

\newpage

\subsection{Limit theory for Network Vector Autoregression}

\subsubsection{Time Series Stationarity}

Below we follow the framework proposed by \cite{zhu2017network} to motivate further our study. Therefore, we begin by assuming that the NVAR model has exogenous regressors which are time-invariant; then the estimation problem reduces to fitting a Vector Autoregression model with network dependence. Within this setting, assume for simplicity that we have the following NVAR(1) model:
\begin{align}
\label{model4}
Y_{i(t)} &= \beta_0 + \beta_1 n_i^{-1} \sum_{j=1}^N \omega_{ij} Y_{j(t-1)} + \beta_2 Y_{i(t-1)} + \boldsymbol{Z}_{i}^{\prime} \boldsymbol{\xi} + \epsilon_{i(t)}
\end{align}
Estimating the NVAR(1) model using the econometric specification \eqref{model3}, implies that the pair $(Y_i,X_i)$ represents a sequence of stationary random variables. In this case, we have node-specific covariates which are time-invariant. This simplification allows to use the conventional central limit theorem under network dependence for the development of the asymptotic theory of the model estimates.   

Let $\mathbb{Z} = \left( Z_1,...,Z_N \right)^{\prime} \in \mathbb{R}^{N \times p}$ and $\boldsymbol{ \mathcal{B} }_0 = \left( \beta_{01},..., \beta_{0N} \right)^{\prime} = \beta_0 \mathbf{1} + \mathbb{Z} \boldsymbol{\xi} \in \mathbb{R}^N$, where $\textbf{1} = \left( 1,...,1 \right)^{\prime}$ is the unit vector of the same dimension, $\boldsymbol{\xi} = \left( \xi_1,..., \xi_p \right)^{\prime} \in \mathbb{R}^p$ and the vector of response variables across the nodes is denoted with $\mathbb{Y}_t = \left( Y_{1t},..., Y_{Nt} \right)^{\prime} \in \mathbb{R}^N$. Then, we rewrite the model \eqref{model3} in the matrix companion form as below
\begin{align}
\label{model5}
\mathbb{Y}_t = \boldsymbol{ \mathcal{B} }_0 +  \boldsymbol{G} \mathbb{Y}_{t-1} + \boldsymbol{ \mathcal{E} }_t
\end{align} 
where $\boldsymbol{G} := \beta_1 \widetilde{ \boldsymbol{\Omega} } + \beta_2 \boldsymbol{I}$, with $\widetilde{ \boldsymbol{\Omega} }  := \mathsf{diag} \left\{ n_1^{-1},..., n_N^{-1} \right\} \otimes \boldsymbol{\Omega}$, is the row-normalized adjacency matrix, $\boldsymbol{I}$ is the identity matrix with compatible dimension and $\boldsymbol{ \mathcal{E} }_t = \left( \epsilon_{1t},..., \epsilon_{Nt}  \right) \in \mathbb{R}^N$ is the innovation vector. We assume the existence of a non-random adjacency matrix, that captures the network dependence, therefore both $\boldsymbol{G}$ and $\widetilde{ \boldsymbol{\Omega} }$ are nonstochastic. However, the matrix of intercepts $\boldsymbol{ \mathcal{B} }_0$ is a stochastic quantity which has to be estimated.

\subsubsection{Strict Stationarity}

We consider the stationary properties of the time series $\mathbb{Y}_t$ for both the cases where the network has a fixed structure, that is $N$ is assumed to be fixed, as well as the case where the network has an unbounded size, such that $N \to \infty$. 

\begin{theorem}
\label{theorem1}
Suppose that $\mathbb{E} \norm{ Z_i } < \infty$ and $N$ is fixed. If $| \beta_1 | + | \beta_2 | < 1$, then there exists a unique strictly stationary solution with a finite first-order moment to the NVAR(1) model \eqref{model3}. The solution has the following form
\begin{align}
\label{solution}
\mathbb{Y}_t = \left( \boldsymbol{I} - \boldsymbol{G} \right)^{-1} \boldsymbol{ \mathcal{B} }_0 + \sum_{j=0}^{\infty} \boldsymbol{G}^j \boldsymbol{ \mathcal{E} }_{t-j}. 
\end{align}
\end{theorem}

\newpage

Assuming the existence of the strictly stationary solution \eqref{sol}, we are interested to obtain its conditional distribution given the nodal information set that characterizes the nodes, denoted with $\mathbb{Z}$. Define with $ \mathbb{E}^{*}(.) = \mathbb{E} \big( .| \mathbb{Z} \big)$ and $ \mathsf{cov}^{*} =  \mathsf{cov}(. | \mathbb{Z} )$. For any integer $h$, we denote the conditional auto-covariance function as $\boldsymbol{\Gamma}(h) = \mathsf{cov}^{*} \left( \mathbb{Y}_t, \mathbb{Y}_{t-h} \right)$. Moreover, it holds that $\boldsymbol{\Gamma}(h) = \boldsymbol{\Gamma}^h (0) \boldsymbol{\Gamma} (0)$ for $h > 0$, and $\boldsymbol{\Gamma}(h) = \boldsymbol{\Gamma}(0) \left( \boldsymbol{\Gamma}^{\prime} \right)^{-h}$ for $h < 0$.  Then, the conditional mean and covariance of $\mathbb{Y}_t$ are obtained by Proposition \ref{proposition1}. 

\begin{proposition}
\label{proposition1}
Assume the same conditions as in Theorem \ref{theorem1}. Then, given $\mathbb{Z}$, the strictly stationary solution of \eqref{solution} converges to a normal distribution with mean and covariance given as below
\begin{align}
\boldsymbol{\mu} &= \left( \boldsymbol{I} - \boldsymbol{G} \right)^{-1} \mathcal{B}_0 = \left( \boldsymbol{I} - \beta_1 \boldsymbol{\Omega} - \beta_2 I \right)^{-1} \mathcal{B}_0
\\
\mathsf{vec} \big\{  \boldsymbol{\Gamma}(0) \big\} &= \sigma^2 \big( \boldsymbol{I} -  \boldsymbol{G} \otimes \boldsymbol{G} \big)^{-1} \mathsf{vec} (  \boldsymbol{I}  ).  
\end{align}
\end{proposition} 
The implications of Proposition \ref{proposition1} is that we can determine the factors that affect the conditional mean of $\mathbb{Y}_t$, which are: (i) the nodal impact $\boldsymbol{ \mathcal{B} }_0$, (ii) the network effect $\beta_1$, (iii) the momentum effect $\beta_2$, and (iv) the network structure\footnote{Note that we assume that the adjacency matrix can take either binary values indicating the existence of a link between the pair $(i,j)$ or represent a weight of the strength of connection.} given by the adjacency matrix $\widetilde{ \boldsymbol{\Omega} }$.

\subsubsection{Parameter estimation}  

The parameters of the NVAR(1) model can be estimated by assuming a Gaussian innovation process. Let $\boldsymbol{\beta} = \left( \beta_0, \beta_1, \beta_2 \right)^{\prime} \in \mathbb{R}^3$ and $\boldsymbol{\theta} = (  \boldsymbol{\theta}_j  )^{\prime} = \big( \boldsymbol{\beta}^{\prime}, \boldsymbol{\xi}^{\prime} \big)^{\prime} \in \mathbb{R}^{p+3}$. To estimate the unknown parameter $\boldsymbol{\theta}$, we rewrite the $\mathsf{NVAR}(1)$ model as below
\begin{align}
\label{model6a}
Y_{i(t)} 
= 
\beta_0 + \beta_1 w_i^{\prime} \mathbb{Y}_{t-1} + \beta_2 Y_{i(t-1)} + \boldsymbol{Z}_i^{\prime} \boldsymbol{\xi} + \epsilon_{i(t)} 
=  
\boldsymbol{ \mathcal{X} }_{i(t-1)}^{\prime} \boldsymbol{\theta} +  \boldsymbol{\epsilon}_{i(t)},
\end{align}  
where $\boldsymbol{ \mathcal{X} }_{i(t-1)} = \left( 1, w_i^{\prime} \mathbb{Y}_{t-1}, Y_{i(t-1)}   , Z_i^{\prime} \right)^{\prime} \in \mathbb{R}^{p+3}$, and $w_i = \left( \omega_{ij} / n_i \right)^{\prime} \in \mathbb{R}^N$ for $1 \leq j \leq N$ is the $i-$ row vector of $W$. 

Moreover, denote with $\mathbb{X}_t = \left( \boldsymbol{ \mathcal{X} }_{1t}, \boldsymbol{ \mathcal{X} }_{2t},..., \boldsymbol{ \mathcal{X} }_{Nt}  \right)^{\prime} \in \mathbb{R}^{N \times (p+3)}$. Then, the NVAR(1) model \eqref{model6a} can be rewritten in vector form $\mathbb{Y}_t = \mathbb{X}_{t-1}  \boldsymbol{\theta} + \boldsymbol{ \mathcal{E} }_t$. Therefore, an ordinary least squares type estimator can be obtained by the following expression 
\begin{align}
\widehat{ \boldsymbol{\theta} }_{ols} = \left( \sum_{t=1}^T \mathbb{X}_{t-1}^{\prime} \mathbb{X}_{t-1} \right)^{-1} \left( \sum_{t=1}^T \mathbb{X}_{t-1}^{\prime} \mathbb{Y}_{t} \right),
\end{align}
whose asymptotic properties are to be investigated subsequently. The following conditions are imposed to allow the development of the asymptotic theory. Within the stationary framework, the asymptotic behaviour of the OLS estimator for the $\mathsf{NVAR}(1)$ model is only affected by the given structure of the model, which implies no presence of nuisance parameters such as the unknown coefficient of persistence.

\newpage

\begin{theorem}
\label{theorem2}
Assume that the stationary condition $| \beta_1 | + | \beta_2 | < 1$ and technical conditions (C1)-(C3) hold, we then have that 
\begin{align*}
\sqrt{NT} \left( \widehat{\theta} - \theta \right) \to \mathcal{N} \left( 0, \sigma^2 \Sigma^{-1} \right)
\end{align*}
as min$\left\{ N, T \right\} \to \infty$, where $\Sigma$.  
\end{theorem}

\medskip

\begin{proposition}
Assume that $T$ is fixed and conditions in Theorem 3 hold. Then, we have that
\begin{align*}
\sqrt{N} \left( \widehat{\theta} - \theta \right) \to \mathcal{N} \left(  0, \sigma^2 T^{-1} \Sigma^{-1} \right)
\end{align*}
as $N \to \infty$. 
\end{proposition}

\medskip

\begin{remark}
For example, recently \cite{fan2023estimation} consider conditions for covariance stationarity of a quantile vector autoregressive model such that 
\begin{align}
\frac{1}{ \sqrt{T} } \sum_{t=1}^T \big( Y_t - \mu_Y \big) \sim \mathcal{N} \left( 0, \underset{ T \to \infty }{ \mathsf{lim}  } \ \sum_{t=1}^T \mathbb{E} \big( Y_t - \mu_Y \big) \big( Y_t - \mu_Y \big)^{\prime} \right).   
\end{align}
\end{remark}
Furthermore, in the absence of network dependence our econometric specification reduces to the predictive regression model  around the general vicinity of unity as in \cite{PM2009econometric}. Usually, in those frameworks the study of vector autoregressive processes of order 1 implies that are integrated of order 1 and cointegrated, that is, processes for which the first difference is stationary and there exists some linear combinations of the coordinate processes that are stationary.

\newpage 

\subsubsection{Time Series Nonstationarity}

We are interested to develop the asymptotic theory for the case where we have time varying regressors with certain nonstationary properties. Therefore, we aim to consider a NVAR model with exogenous nonstationary regressors under a known network structure. In our case, we are interested to examine the conditional distribution given the nonstationary regressors of the nodes.

Consider again the model specification with the LUR process as below
\begin{align}
\label{model21}
Y_{i(t)} &= \alpha_0 + \alpha_1 n_i^{-1} \sum_{j=1}^N \omega_{ij} Y_{j(t-1)} + \alpha_2 Y_{i(t-1)} + B X_{i(t-1)}  + \epsilon_{i(t)}
\\
\label{model22}
X_{i(t)} &= \boldsymbol{R} X_{i(t-1)}  + U_{i(t)}
\end{align}
where $\boldsymbol{R}_{T} = \left( \boldsymbol{I}_p + \frac{ \boldsymbol{C}_p }{T^{\lambda }} \right)$, the autocorrelation coefficient matrix, $\boldsymbol{C}_p = \mathsf{diag} \left\{ c_1,..., c_p \right\}$ with $c_i$ the unknown persistence coefficient, $\lambda < 1, \lambda \in (0,1)$ or $\lambda > 1$, the exponent rate and $\alpha_0, \alpha_1, \alpha_2$ and $B = \text{diag} \left\{ \beta_1,..., \beta_p \right\}$   are the model parameters to be estimated.

\subsubsection{Parameter estimation}

Denote with $\widetilde{X}_{i(t-1)} = \left( 1, w_i^{\prime} \mathbb{Y}_{t-1}, Y_{i(t-1)}, \widetilde{Z}_{i(t)}^{\prime} \right)^{\prime} \in \mathbb{R}^{p+3}$, and $w_i = \left( \omega_{ij} / n_i \right)^{\prime} \in \mathbb{R}^N$ for $1 \leq j \leq N$ is the $i-$ row vector of $W$. Moreover, denote with $\widetilde{\mathbb{X}}_t = \left( \widetilde{X}_{1t}, \widetilde{X}_{2t},..., \widetilde{X}_{Nt} \right)^{\prime} \in \mathbb{R}^{N \times (p+3)}$ the vector composing the design matrix that aligns with $T$ time observations. Notice also that in the case where the NVAR model is expressed in a similar manner as the predictive regression system the design matrix includes the IVX instrument which is mildly integrated version of the original time varying regressor. 

Then, the NVAR(1) model \eqref{model21} can be rewritten in vector form $\mathbb{Y}_t = \widetilde{\mathbb{X}}_{t-1} \widetilde{\theta} + \mathcal{E}_t$. Therefore, the IVX type estimator of the NVAR(1) model can be obtained by
\begin{align}
\widetilde{\theta}_{\text{IVX}} = \left( \sum_{t=1}^T \widetilde{\mathbb{X}}_{t-1}^{\prime} \widetilde{\mathbb{X}}_{t-1} \right)^{-1} \sum_{t=1}^T \widetilde{\mathbb{X}}_{t-1}^{\prime} \mathbb{Y}_{t},
\end{align}

We aim to show that the asymptotic distribution of the $\widetilde{\theta}_{\text{IVX}}$ is mixed Gaussian and to determine the stochastic variance term of its limiting distribution. 
\begin{theorem}
\label{theorem3}
Under Assumption 2 (innovation covariance structure), we then have that
\begin{align*}
\sqrt{NT} \big( \widehat{\widetilde{\theta}}_{\text{IVX}} - \widetilde{\theta}_{\text{IVX}} \big) \to \mathcal{MN} \left( 0, \sigma^2 \mathbb{V}^{-1} \right)
\end{align*}
as min$\left\{ N, T \right\} \to \infty$, where $\mathbb{V}$ a positive covariance matrix.   
\end{theorem}

\newpage 

\subsection{General NVAR(p) Model}


In the previous sections, we consider the NVAR(1) model, however we can easily extend the model in the case where $p > 1$ while concentrating on the nonstationary case. Therefore, the  NVAR$(p)$ model is expressed as below:
\begin{align}
\label{model5}
Y_{i(t)} &= \beta_0 + \sum_{s=1}^p \alpha_{s} n_i^{-1} \sum_{j=1}^N \omega_{ij} Y_{j(t-s)} + \sum_{j=1}^p \beta_{s} Y_{i(t-s)} + \Xi X_{i(t-1)} + \epsilon_{i(t)}
\\
X_{i(t)} &= \mathcal{R}_{T} X_{i(t-1)}  + U_{i(t)}
\end{align}
where $\mathcal{R}_{T} = \left( I_p + \frac{C_p}{T^{\lambda }} \right)$, with $\lambda < 1, \lambda \in (0,1)$ or $\lambda > 1$ and $C_p = \text{diag} \left\{ c_1,..., c_p \right\}$ and $\Xi = \text{diag} \left\{ \xi_1,...,\xi_p \right\}$ is the coefficient matrix for the vector of regressors. Moreover, we define with $\mathbb{Y}_t = \left( \mathbb{Y}_t^{\prime}, \mathbb{Y}_{t-1}^{\prime}, ...., \mathbb{Y}_{t-p+1}^{\prime} \right)^{\prime} \in \mathbb{R}^{Np}$. Then, we express the NVAR$(p)$ model in the the matrix companion form as below
\begin{align}
\label{model6}
\mathbb{Y}_t^{*} = \mathcal{B}_0^{*} + G^{*} \mathbb{Y}_{t-1}^{*} + \mathcal{E}^{*}_{t},
\end{align}
with $\mathcal{B}_0^{*} = \left( \mathcal{B}_0^{\prime}, \mathbf{0}^{\prime}_{N(p-1)} \right) \in \mathbb{R}^{Np}$, $\mathcal{E}^{*}_{t} = \left( \mathcal{E}_t^{\prime}, \mathbf{0}^{\prime}_{N(p-1)} \right)^{\prime} \in \mathbb{R}^{Np}$ where $G^{*}$ is defined below 
\begin{align}
G^{*} = 
\begin{pmatrix}
\mathcal{K} \ \ & \ \ \alpha_p W + \beta_p I_N \\
I_{N(p-1)} \ \ & \ \  \mathbf{0}_{N(p-1),N}    \\
\end{pmatrix},
\end{align} 
such that $\mathcal{K} = \left( \alpha_{1} W + \beta_1 I_N ,...., \alpha_{p-1} W + \beta_{p-1} I_N  \right) \in \mathbb{R}^{N \times N(p-1)}$, $\mathbf{0}_{n}$ the n-dimensional zero vector and $I_n$ the identity matrix. Furthermore, in the case that we replace the vector of time-varying regressors $X_{i(t-1)}$ with $Z_i$ some time-invariant node-specific characteristic, then we can consider the strictly stationary solution of the NVAR$(p)$ model as we presented with Theorem \ref{theorem1} above in the case of the NVAR(1) model. 

\begin{theorem}
\label{theorem4}
Consider the case where $X_{i(t-1)} \equiv Z_i$. Assume that $\textbf{E} \norm{ Z_i } < \infty$ and $N$ is fixed. If $\sum_{s=1}^p \left( | \alpha_s | + | \beta_s | \right) < 1$, then there exists a unique strictly stationary solution with a finite first-order moment to the NVAR(p) model \eqref{model5} of the form:
\begin{align}
\label{sol}
\mathbb{Y}_t = \mathcal{J} \mathbb{Y}_t^{*}, \ \ \text{where} \ \   \mathbb{Y}_t^{*} = \left( I_p - G^{*} \right)^{-1} \mathcal{B}_0^{*} + \sum_{j=0}^{\infty} G^{*j} \mathcal{E}_{t-j}^{*}
\end{align}
\end{theorem}
such that the following hold
\begin{align*}
\mathcal{J} \left( I - G^{*} \right)^{-1} \mathcal{B}_0^{*} = \left( I - \widetilde{G} \right)^{-1} \mathcal{B}_0, \ \widetilde{G} = \sum_{s=1}^p \left( \alpha_s W + \beta_s I_N  \right) \ \text{and} \ \mathcal{J} = \left[ I_N \ \mathbf{0}_{N(p-1),N} \right].
\end{align*}

\subsubsection{Parameter Estimation for Stationary framework}

Assume that the node-specific covariates $Z_i$ is a $d-$dimensional vector. Then, we write $\mathcal{X}^{*}_{i(t-1)} = \left( 1, w_i^{\prime} \mathbb{Y}_{t-1},..., \mathbb{Y}_{t-p}, Y_{i(t-p)}, Z_i^{\prime} \right)^{\prime} \in \mathbb{R}^{2p + d + 1}$ and $\mathbb{X}^{*}_{t-1} = \left( \mathcal{X}^{*}_{1(t-1)},..., \mathcal{X}^{*}_{N(t-1)} \right) \in \mathbb{R}^{N \times (2p + d + 1)}$. Moreover, denote the parameter vector with $\theta^{*} = \left( \beta_0, \alpha^{\prime}, \beta^{\prime}, \xi^{\prime} \right)^{\prime} \in \mathbb{R}^{2p + d + 1}$, where $\alpha = \left( \alpha_1,..., \alpha_p \right)^{\prime}$ and $\beta = \left( \beta_1,..., \beta_p \right)^{\prime}$. Then, model \eqref{model6} can be written as $\mathbb{Y}_t = \mathbb{X}^{*}_{t-1} \theta^{*} + \mathcal{E}_t$. Then, the ordinary least squares type estimator can be obtained by
\begin{align}
\widehat{\theta}_{\text{OLS}}^{*} = \left( \sum_{t= p+1}^T \mathbb{X}_{t-1}^{* \prime} \mathbb{X}^{*}_{t-1} \right)^{-1} \sum_{t= p+1}^T  \mathbb{X}_{t-1}^{*\prime} \mathbb{Y}_{t},
\end{align}

In order to investigate the asymptotic properties of $\widehat{\theta}_{\text{OLS}}^{*}$, we define $\Gamma^{*}(h) = \text{cov}^{*} \left( \mathbb{Y}_t, \mathbb{Y}_{t-h} \right)$ to be the conditional auto-covariance function for the NVAR$(p)$ model under the assumption of strict stationarity. Theorem \ref{theorem5} gives the limiting distribution of the estimator. 

\medskip

\begin{theorem}
\label{theorem5}
Assume that $\sum_{s=1}^p \left( | \alpha_s | + | \beta_s | \right) < 1$ and technical conditions (C1),(C2), (C4) hold. We then have that 
\begin{align*}
\sqrt{NT} \left( \widehat{\theta}^{*}_{\text{OLS}} - \theta^{*}_{\text{OLS}} \right) \to \mathcal{N} \left( 0, \sigma^2 \Sigma^{*-1} \right)
\end{align*}
as min$\left\{ N, T \right\} \to \infty$, where $\Sigma^{*}$. 
\end{theorem}

\medskip

\newpage 

\subsubsection{Parameter Estimation for Nonstationary framework}

Under the assumptions we impose for the nonstationary framework, we utilize the models \eqref{model5}-\eqref{model6} and in particularly we assume the existence of a vector of nonstationary regressors generated by the LUR specification. Similarly, as in the previous section we aim to investigate the asymptotic distribution of the IVX estimator for the NVAR$(p)$.   

Denote with $\widetilde{X}_{i(t-1)}^{*} = \left( 1, w_i^{\prime} \mathbb{Y}_{t-1},..., \mathbb{Y}_{t-p}, Y_{i(t-p)}, \widetilde{Z}_i^{\prime} \right)^{\prime} \in \mathbb{R}^{2p + d + 1}$ where the elements $w_i = \left( \omega_{ij} / n_i \right)^{\prime} \in \mathbb{R}^N$ for $1 \leq j \leq N$ represent the $i-$ row vector of $W$. Moreover, denote with 
\begin{align*}
\widetilde{\mathbb{X}}^{*}_{t-1} = \left( \widetilde{X}^{*}_{1(t-1)}, \widetilde{X}^{*}_{2(t-1)},..., \widetilde{X}^{*}_{N(t-1)} \right)^{\prime} \in \mathbb{R}^{N \times (2p + d + 1)},
\end{align*}
the vector that corresponds to the design matrix of the model that aligns with $T$ time observations. We also, denote the parameter vector with $\theta^{*} = \left( \beta_0, \alpha^{\prime}, \beta^{\prime}, \Xi \right)^{\prime} \in \mathbb{R}^{2p + d + 1}$, where $\alpha = \left( \alpha_1,..., \alpha_p \right)^{\prime}$ and $\beta = \left( \beta_1,..., \beta_p \right)^{\prime}$. Therefore, the NVAR$(p)$ model \eqref{model4} within the nostationary framework can be rewritten in vector form $\mathbb{Y}_t = \widetilde{\mathbb{X}}^{*}_{t-1} \widetilde{\theta}^{*} + \mathcal{E}_t$. Therefore, the IVX type estimator of the NVAR(1) model can be obtained by
\begin{align}
\widehat{ \widetilde{\theta} }^{*}_{\text{IVX}} = \left( \sum_{t=1}^T \widetilde{\mathbb{X}}_{t-1}^{*\prime} \widetilde{\mathbb{X}}^{*}_{t-1} \right)^{-1} \sum_{t=1}^T \widetilde{\mathbb{X}}_{t-1}^{*\prime} \mathbb{Y}_{t},
\end{align}

\begin{theorem}
\label{theorem6}
Under Assumption 2 (innovation covariance structure), we then have that 
\begin{align*}
\sqrt{NT} \big( \widehat{ \widetilde{\theta}}^{*}_{\text{IVX}} - \widetilde{\theta}^{*}_{\text{IVX}} \big) \to \mathcal{MN} \left( 0, \sigma^2 \mathbb{V}^{*-1} \right)
\end{align*}
as min$\left\{ N, T \right\} \to \infty$, where $\mathbb{V}^{*}$ is defined to be...   
\end{theorem}
The proof of this theorem should be the second theoretical contribution of the paper to the literature. The proof should be similar to the proof of Theorem \ref{theorem3}.

\newpage 

\section{Asymptotic Theory}
\label{section3}

\subsection{Assumptions on Network Dependence}

To develop the asymptotic theory of our framework we shall discuss some relevant assumptions and limit theorems from the network analysis perspective.

\subsubsection{Asymptotic Uncorrelation}

Let $G_n = \left( V_n, E_n \right)$ be an undirected network and let $\left( Z_i \right)_{ i \in G_n }$ be a set of random variables indexed by the edges. We assume that $\left( Z_i \right)_{ i \in G_n }$ is jointly exchangeable in the vertices, that is, $\sigma: V_n \to V_n$ be a permutation of the vertex set. For an edge $i = (v_1, v_2) \in G_n$, we denote by $\sigma_{(i)} := \left( \sigma_{ (v_1)},  \sigma_{ (v_2) } \right)$ the permuted edge. The network is called exchangeable if $\left( Z_i \right)_{ i \in G_n }$ and $\left( Z_{ \sigma(i) } \right)_{ i \in G_n }$ have the same joint distribution for all permutations of the vertex set $\sigma$. In particular, this means that $Z_i$ and $Z_j$ are identically distributed if there is a permutation $\sigma$ such that $i = \sigma(j)$. Notice that this is always the case because we have defined networks as having no loops, that is, for $i = ( v, v^{\prime} )$ we always assume that $v \neq v^{\prime}$. Next, we want to describe under which circumstances two pairs $\left(  Z_{ i_1 }, Z_{ j_1 } \right)$ and  $\left(  Z_{ i_2 }, Z_{ j_2 } \right)$ for $i_1, j_1, i_2, j_2 \in G_n$ have the same distribution. Denote with $\kappa(i,j) := | e_i \cap e_j | \in \left\{ 0,1,2 \right\}$ be the number of common vertices of $i$ and $j$. The following lemma is easy to prove. 

\begin{lemma}
Let $i_1, j_1, i_2, j_2 \in G_n$. There is a permutation of the vertices $\sigma$ such that $( i_1, j_1 ) = \left( \sigma(i_2) , \sigma(j_2)   \right)$ if and only if $\kappa(i_1, j_1) = \kappa(i_2, j_2)$. 
\end{lemma}

\begin{proof}
Let $i = ( v_{i_1}, v^{\prime}_{i_1} )$ and analogously for $i_2, j_1$ and $j_2$. Let $\sigma$ be a permutation such that $( i_1, j_1 ) = \left( \sigma( i_2 ), \sigma( j_2 ) \right)$. Then, we have that
\begin{align*}
\kappa (i_1, j_1 ) = \kappa \left( \sigma( i_2 ), \sigma( j_2 )  \right) = \big| \left\{ \sigma(v_{i_2} ),   \sigma(v^{\prime}_{i_2} )  \right\}  \cap    \left\{ \sigma(v_{j_2} ),   \sigma(v^{\prime}_{j_2} )  \right\} \big| = | e_{i_2} \cap  e_{j_2} | = \kappa (i_2, j_2 ). 
\end{align*}  
Thus, if $\kappa(i_1, j_1) =  \kappa(i_2, j_2)$ we can easily construct $\sigma$ with $(i_1, j_1) = \left( \sigma(i_2), \sigma(j_2) \right)$ just by mapping the corresponding vertices onto each other and letting the other vertices unchanged. 
\end{proof}

Next, we obtain the following corollary which characterizes when two pairs of random variables are identically distributed. 

\begin{corollary}
Let $\left( Z_{ \sigma(i) } \right)_{ i \in G_n }$ be a sequence of exchangeable random variables indexed by the edges of a network $G_n$. For any vertices $i_1, j_1, i_2, j_2 \in G_n$, we have that $\left( Z_{i_1},  Z_{j_1} \right)$ and $\left( Z_{i_2},  Z_{j_2} \right)$ are identically distributed if  $\kappa(i_1, j_1) = \kappa(i_2, j_2)$. 
\end{corollary}

\newpage

\begin{corollary}
For all $n \in \mathbb{N}$, let $G_n = ( V_n, E_n )$ be undirected and complete networks and assume that $\left( Z_{ \sigma(i) } \right)_{ i \in G_n }$ are interchangeable and square integrable. Recall that $r_n = | G_n | = \frac{n(n-1)}{2}$ is the number of edges. Then, for pairwise different vertices $v_1, v_2, v_3, v_4 \in V_n$, 
\begin{align*}
Var \left( \frac{1}{v_n} \sum_{i \in G_n } Z_{n,i} \right) 
&=  \frac{1}{v_n^2} \sum_{i \in G_n } \text{Var} \left( Z_{n,i} \right) + \frac{1}{v_n^2} \sum_{ \substack{ i, j \in G_n \\ \kappa(i,j) = 1} } Cov \left( Z_{n,i}, Z_{n,j} \right)  + \frac{1}{v_n^2} \sum_{ \substack{ i, j \in G_n \\ \kappa(i,j) = 0} } Cov \left( Z_{n,i}, Z_{n,j} \right)
\\
&= r_n^{-1} Var \left( Z_{n, v_1v_2 } \right)  + \mathcal{O} \left( r_n^{ -1/2} \right) Cov \left( Z_{n,v_1v_2}, Z_{n,v_2v_3} \right) + \mathcal{O}(1) Cov \left( Z_{n,v_1v_2}, Z_{n,v_3v_4} \right) 
\end{align*}
\end{corollary}

\begin{remark}
For instance, if we assume that $Z_{n,i}$ and $Z_{n,j}$ are uncorrelated when $i \neq j$, the covariances vanish and we see that $Var \left( \frac{1}{v_n} \sum_{i \in G_n } Z_{n,i} \right) \to 0$ as $n \to \infty$ if   
$r_n^{-1} Var \left( Z_{n,v_1v_2} \right) \to 0$  as $n \to \infty$. Notice that as we have motivated in the previous corollary we do not need asymptotic uncorrelation for all edges, we need it only for edges $i$ and $j$ with $\kappa(i,j) = 0$. Furthermore, for edges $i$ and $j$ with $\kappa(i,j) \leq 1$ we merely need that the covariances do not grow too fast. Thus, we make an assumption on the average behaviour of disjoint edges, For  $\kappa(i,j) = 0$, we argue that $Cov \left( Z_{n,v_1v_2}, Z_{n,v_3v_4} \right) \to 0$ is a reasonable assumption, because we believ that most edges are separated and do not strongly influence each other.  
\end{remark}

\medskip

Consequently, to correctly establish the asymptotic theory within our proposed econometric environment  we shall consider the implementation of suitable stochastic approximations for functionals of Brownian motion under the presence of network (graph) dependence. We discuss two relevant studies, namely the framework proposed by \cite{laurent2022unit} as well as the framework of \cite{gobet2017parameter}. In particular, the latter considers parameter estimation of Orneisten-Uhlenbeck processes generating a stochastic graph while the former considers the development of unit root testing for high-dimensional data.  Notice that there are several issues we need to make sure are clarified. In particular, a more complex scenario would be the case we allow for an increasing sequence of graphs which have some form of temporal dependence. On the other hand, to simplify our setting we assume that we have a fix time length $T$ and consider a random graph $A_T$. For instance, if we would allow for the evolution of the network through time then we will need to impose related conditions on the correlation structure between graphs at fixed time length increments. 

Suppose that we have a set of stochastic graphs such as the adjacency value between vertices $i$ and $j$ is given by $\boldsymbol{\Omega}_{ij} := \boldsymbol{1} \left\{  \boldsymbol{J}_C (t) \in S_{ij} \right\}$. In other words, we can consider those network-dependent covariates as the jump-components in a time series regression for high dimensional data formulated to represent a stochastic process with unit roots. There is also a growing literature which considers unstable networks when modeling interactions of economic agents as in the study of \cite{badev2021nash}.

\newpage

\begin{wrap}
\subsection{Quasi-Maximum Likelihood Estimation Approach}

Estimation for the unknown parameter vector $\theta$ is developed by means of QMLE. We define the quasi-log-likelihood function for $\boldsymbol{\theta}$ as below
\begin{align}
\ell_{NT} (\boldsymbol{\theta}) = \sum_{t=1}^T \sum_{i=1}^N \ell_{i,t} (\boldsymbol{\theta}),    
\end{align}
where $\ell_{i,t} (\boldsymbol{\theta})$ is the log-likelihood contribution of a single network node which depends on the data, that is on the persistence properties of regressors in our case. 

We consider the following QMLE estimator 
\begin{align}
\ell_{NT} (\boldsymbol{\theta}) = \sum_{t=1}^T \sum_{i=1}^N \big( Y_{i,t} \mathsf{log} \lambda_{i,t}(\boldsymbol{\theta}) - \lambda_{i,i}(\boldsymbol{\theta}) \big),     
\end{align}
which is the log-likelihood obtained if all time series were contemporaneously independent. 

The particular property simplifies computations allowing to establish consistency and asymptotic normality of the resulting estimator. Some of the challenges include the fact that non standard proofs are required for establishing stationarity or infinite-dimensional processes in order to establish robust inference procedures within the double asymptotic regime. Suppose that the functional form of the model is described by $\lambda_t ( \theta)$, then $\hat{\theta}$ corresponds to the maximizer of the least squares criterion such that 
\begin{align}
\ell_{NT} (\theta) = - \sum_{t=1}^{T} \big( Y_t - \lambda_t ( \boldsymbol{\theta} ) \big)^{\prime} \big( Y_t - \lambda_t ( \boldsymbol{\theta} ) \big).    
\end{align}
Then, it follows that 
\begin{align}
\mathcal{S}_{NT}(\theta) = \sum_{t=1}^T \frac{ \partial \lambda_t (\boldsymbol{\theta}) }{ \partial \boldsymbol{\theta} } \big( Y_t - \lambda_t ( \boldsymbol{\theta} ) \big) = \sum_{t=1}^T s_{NT} (\boldsymbol{\theta}).      
\end{align}
Furthermore, the empirical Hessian and information matrices are respectively as below
\begin{align}
H_{NT}(\boldsymbol{\theta}) 
&= 
\sum_{t=1}^T \sum_{i=1}^N \frac{ \partial \lambda_{i,t} (\boldsymbol{\theta}) }{ \partial \boldsymbol{\theta} } \cdot \frac{ \partial \lambda_{i,t} (\boldsymbol{\theta}) }{ \partial \boldsymbol{\theta}^{\prime} } - \sum_{t=1}^T \sum_{i=1}^N \big( Y_{i,t} - \lambda_{i,t} ( \boldsymbol{\theta} ) \big) \frac{ \partial^2 \lambda_{i,t} (\boldsymbol{\theta}) }{ \partial \boldsymbol{\theta} \partial \boldsymbol{\theta}^{\prime} }
\\
\nonumber
\\
B_{NT}(\boldsymbol{\theta}) 
&= 
\sum_{t=1}^T \frac{ \partial \lambda_{t}^{\prime} (\boldsymbol{\theta}) }{ \partial \boldsymbol{\theta} } \boldsymbol{\Sigma}_t (\boldsymbol{\theta}) \frac{ \partial \lambda_{t}^{\prime} (\boldsymbol{\theta}) }{ \partial \boldsymbol{\theta}^{\prime} }. 
\end{align}
Moreover, we have that 
\begin{align}
H_{N}(\theta) &= \mathbb{E} \left( \frac{ \partial \lambda_{t}^{\prime} (\boldsymbol{\theta}) }{ \partial \boldsymbol{\theta} } \cdot  \frac{ \partial \lambda_{t} (\boldsymbol{\theta}) }{ \partial \boldsymbol{\theta}^{\prime} }  \right)
\\
B_N &= \mathbb{E} \left( \frac{ \partial \lambda_{t}^{\prime} (\boldsymbol{\theta}) }{ \partial \boldsymbol{\theta} } \cdot \xi_t(\theta) \cdot \xi^{\prime}_t(\theta)  \frac{ \partial \lambda_{t} (\boldsymbol{\theta}) }{ \partial \boldsymbol{\theta}^{\prime} }  \right)  
\end{align}

There, there exists a fixed open neighborhood $\mathcal{D} (\theta_0) = \big\{ \theta : \left| \theta - \theta \right|_2 \big\}$ of $\theta_0$ such that with probability tending to 1 as $\left\{ N, T_N \right\} \to \infty$, then the score function equation $\mathcal{S}_{NT} (\theta) = 0$, has a unique solution which is denoted by $\hat{\theta} \overset{p}{\to} \theta_0$ and $\sqrt{ N T_N } \left( \hat{\theta} - \theta_0 \right) \overset{d}{\to} \mathcal{N} \left( 0, B^{-1}  \right)$, where $B = \sigma^2 H$ and $H$ is defined, with $\Sigma_t = D_t = I$.

\end{wrap}


\newpage 

\begin{wrap}

\subsection{Theoretical Proofs}
\label{Section3.1}

\paragraph{Proof of Theorem 2 of \cite{zhu2017network}}
\
\color{blue}

Write the estimator of $\widehat{\theta}$ such that $\widehat{\theta} = \theta + \widehat{\Sigma}^{-1} \widehat{\Sigma}_{xe}$, where  
\begin{align}
\widehat{\Sigma} = (NT)^{-1} \sum_{t=1}^T \mathbb{X}_{t-1} \mathbb{X}_{t-1} ^{\prime} \ \ \text{and} \ \ \widehat{\Sigma}_{xe} = (NT)^{-1} \sum_{t=1}^T \mathbb{X}_{t-1} \mathcal{E}_{t}
\end{align}
As a result the conclusion of Theorem 3 holds if, 
\begin{align}
\label{result1}
\widehat{\Sigma} &\to \Sigma \\
\label{result2}
\sqrt{NT} \widehat{\Sigma}_{xe} &\to \mathcal{N} \left( 0 , \sigma^2 \Sigma \right)
\end{align}
as min$\left\{ N, T \right\} \to \infty$.  

Below, we prove \eqref{result1} in Step 1 and \eqref{result2} in Step 2. 

\underline{Step 1.} In this step, we prove the following result
\begin{align}
\widehat{\Sigma} = \frac{1}{NT} \sum_{t=1}^T \mathbb{X}_{t-1}^{\prime} \mathbb{X}_{t-1} 
=
\begin{pmatrix}
1 & \mathcal{S}_{12} & \mathcal{S}_{13} & \mathcal{S}_{14} \\
 & \mathcal{S}_{22} & \mathcal{S}_{23} & \mathcal{S}_{24} \\
 &  & \mathcal{S}_{33} & \mathcal{S}_{34} \\
 &  &  & \mathcal{S}_{44} \\
\end{pmatrix}
\to 
\begin{pmatrix}
1 & c_{\beta} & c_{\beta} & 0 \\
 & \Sigma_1 & \Sigma_2 & \kappa_3 \gamma^{\prime} \Sigma_z  \\
 &  & \Sigma_3 & \kappa_8 \gamma^{\prime} \Sigma_z  \\
 &  &  & \Sigma_z \\
\end{pmatrix}
\end{align}
where  
\begin{align}
\mathcal{S}_{12} &= \frac{1}{NT} \sum_{t=1}^T \sum_{i=1}^N w_i^{\prime} \mathbb{Y}_{t-1}, \ \ \mathcal{S}_{13} = \frac{1}{NT} \sum_{t=1}^T \sum_{i=1}^N Y_{i(t-1)}, \ \ \mathcal{S}_{14} = \frac{1}{N} \sum_{i=1}^N Z_{i}^{\prime} 
\nonumber
\\
\nonumber
\\
\mathcal{S}_{22} &= \frac{1}{NT} \sum_{t=1}^T \sum_{i=1}^N \left( w_i^{\prime} \mathbb{Y}_{t-1} \right)^2, \ \ \mathcal{S}_{23} = \frac{1}{NT} \sum_{t=1}^T \sum_{i=1}^N w_i^{\prime} \mathbb{Y}_{t-1} Y_{i(t-1)}
\nonumber
\\
\nonumber
\\
\mathcal{S}_{24} &= \frac{1}{NT} \sum_{t=1}^T \sum_{i=1}^N  w_i^{\prime} \mathbb{Y}_{t-1} Z_i^{\prime} , \ \ \mathcal{S}_{33} = \frac{1}{NT} \sum_{t=1}^T \sum_{i=1}^N Y^2_{i(t-1)}
\nonumber
\\
\nonumber
\\
\mathcal{S}_{34} &= \frac{1}{NT} \sum_{t=1}^T \sum_{i=1}^N Y_{i(t-1)} Z_i^{\prime}, \ \ \mathcal{S}_{44} = \frac{1}{N} \sum_{i=1}^N Z_i Z_i^{\prime}
\end{align}


\newpage 

Therefore, we have that 
\begin{align}
\mathbb{Y}_t = c_{\beta} \mathbf{1} + (I-G)^{-1} \mathbb{Z} \gamma + \widetilde{\mathbb{Y}}_t,
\end{align}
almost surely. 
where $\widetilde{\mathbb{Y}}_t = \sum_{j=0}^{\infty} G^j \mathcal{E}_{t-j}$, $\mathbb{Z} \gamma = \left( Z_1^{\prime}, ... ,  Z_N^{\prime} \right)^{\prime}$

For the case where no information regarding the persistence properties of regressors (that is, we assume that $X_{i(t)}$ are not generated via the LUR specification), then we see that the limiting distribution of the model estimate for the NVAR(1) model weakly converges to a Normal random variable.

\end{wrap}


\subsection{Nonstationary Framework}

In order to establish the asymptotic theory in our framework we shall establish the weak convergence of functionls of Brownian motion in graphs into OU processes under graph dependence. In this Section we focus on the proofs of the related theorems that give the asymptotic distribution of the estimators based on the IVX instrumentation. We follow similar derivations as in Section \ref{Section3.1}, but for the nonstationary framework we assume that the vector of regressors $X_{i(t)}$ is generated by the LUR process.

Let $\widetilde{X}_{i(t-1)} = \left( 1, w_i^{\prime} \mathbb{Y}_{t-1}, Y_{i(t-1)}, \widetilde{Z}_{i(t)}^{\prime} \right)^{\prime} \in \mathbb{R}^{p+3}$, and $w_i = \left( \omega_{ij} / n_i \right)^{\prime} \in \mathbb{R}^N$ for $1 \leq j \leq N$ is the $i-$ row vector of $W$. Moreover, denote with $\widetilde{\mathbb{X}}_t = \left( \widetilde{X}_{1t}, \widetilde{X}_{2t},..., \widetilde{X}_{Nt} \right)^{\prime} \in \mathbb{R}^{N \times (p+3)}$. Then, the NVAR(1) model \eqref{model4} can be rewritten in vector form $\mathbb{Y}_t = \widetilde{\mathbb{X}}_{t-1} \widetilde{\theta} + \mathcal{E}_t$. Therefore, the IVX type estimator of the NVAR(1) model can be obtained by
\begin{align*}
\widetilde{\theta}_{\text{IVX}} = \left( \sum_{t=1}^T \widetilde{\mathbb{X}}_{t-1}^{\prime} \widetilde{\mathbb{X}}_{t-1} \right)^{-1} \sum_{t=1}^T \widetilde{\mathbb{X}}_{t-1}^{\prime} \mathbb{Y}_{t},
\end{align*}

\paragraph{Sketch Proof of Theorem 3}

We define with 
\begin{align}
\widehat{\Sigma} := \frac{1}{NT} \sum_{t=1}^T \widetilde{\mathbb{X}}_{t-1} \widetilde{\mathbb{X}}_{t-1} ^{\prime} \ \ \text{and} \ \ \widehat{\Sigma}_{xe} = \frac{1}{NT} \sum_{t=1}^T \widetilde{\mathbb{X}}_{t-1} \mathcal{E}_{t}
\end{align}
We have that 
\begin{align}
\widehat{\Sigma} = \frac{1}{NT} \sum_{t=1}^T \widetilde{\mathbb{X}}_{t-1} \widetilde{\mathbb{X}}_{t-1} ^{\prime} 
= 
\begin{bmatrix}
1 & \mathcal{S}_{12} & \mathcal{S}_{13} & \mathcal{S}_{14} \\
 & \mathcal{S}_{22} & \mathcal{S}_{23} & \mathcal{S}_{24} \\
 &  & \mathcal{S}_{33} & \mathcal{S}_{34} \\
 &  &  & \mathcal{S}_{44} \\
\end{bmatrix}
\end{align}
where  
\begin{align}
\mathcal{S}_{12} &= \frac{1}{NT} \sum_{t=1}^T \sum_{i=1}^N w_i^{\prime} \mathbb{Y}_{t-1}, \ \ \mathcal{S}_{13} = \frac{1}{NT} \sum_{t=1}^T \sum_{i=1}^N Y_{i(t-1)}, \ \ \mathcal{S}_{14} = \frac{1}{N} \sum_{t=1}^T \widetilde{Z}_{t}^{\prime} 
\nonumber
\\
\nonumber
\\
\mathcal{S}_{22} &= \frac{1}{NT} \sum_{t=1}^T \sum_{i=1}^N \left( w_i^{\prime} \mathbb{Y}_{t-1} \right)^2, \ \ \mathcal{S}_{23} = \frac{1}{NT} \sum_{t=1}^T \sum_{i=1}^N w_i^{\prime} \mathbb{Y}_{t-1} Y_{i(t-1)}
\nonumber
\\
\nonumber
\\
\mathcal{S}_{24} &= \frac{1}{NT} \sum_{t=1}^T \sum_{i=1}^N  w_i^{\prime} \mathbb{Y}_{t-1} \widetilde{Z}_t^{\prime} , \ \ \mathcal{S}_{33} = \frac{1}{NT} \sum_{t=1}^T \sum_{i=1}^N Y^2_{i(t-1)}
\nonumber
\\
\nonumber
\\
\mathcal{S}_{34} &= \frac{1}{NT} \sum_{t=1}^T \sum_{i=1}^N Y_{i(t-1)} \widetilde{Z}_t^{\prime}, \ \ \mathcal{S}_{44} = \frac{1}{T} \sum_{t=1}^T \widetilde{Z}_t \widetilde{Z}_t^{\prime}
\end{align}
where $\widetilde{Z}_t$ denotes the IVX instrument. 

Therefore, we need to examine the weakly convergence for each of the above terms in order to determine the covariance of the limiting distribution of the estimator $\widetilde{\theta}_{\text{IVX}}$. 
\color{red}
\underline{Convergence of $\mathcal{S}_{12}$:}
\color{black}
\begin{align}
\mathcal{S}_{12} = \frac{1}{NT} \sum_{t=1}^T \sum_{i=1}^N w_i^{\prime} \mathbb{Y}_{t-1} = \frac{1}{NT} \sum_{t=1}^T \mathbf{1}^{\prime} W \mathbb{Y}_{t-1} = 
\frac{\beta_0}{1 - \beta_1 - \beta_2} + \mathcal{S}_{12}^{A} + \mathcal{S}_{12}^{B}
\end{align}
where 
\begin{align}
\mathcal{S}_{12}^{A} = \frac{1}{N} \mathbf{1}^{\prime} W \left( I - G \right)^{-1} \mathbb{Z} \Xi \ \ \text{and} \ \ \mathcal{S}_{12}^{B} = \frac{1}{NT} \sum_{t=1}^T \mathbf{1}^{\prime} W \widetilde{ \mathbb{Y}}_{t-1}
\end{align}
\color{red}
To check whether the above terms converge to zero, that is, $\mathcal{S}_{12}^{A} \to 0$ and $\mathcal{S}_{12}^{B} \to 0$.
\color{black}

\color{red}
\underline{Convergence of $\mathcal{S}_{13}$:}
\color{black}

\begin{align}
\mathcal{S}_{13} = \frac{1}{NT} \sum_{t=1}^T \sum_{i=1}^N Y_{i(t-1)} = \frac{1}{NT} \sum_{t=1}^T \mathbf{1}^{\prime} \mathbb{Y}_{t-1} = 
\frac{\beta_0}{1 - \beta_1 - \beta_2} + \mathcal{S}_{13}^{A} + \mathcal{S}_{13}^{B}
\end{align}
where 
\begin{align}
\mathcal{S}_{13}^{A} = \frac{1}{N} \mathbf{1}^{\prime} \left( I - G \right)^{-1} \mathbb{Z} \Xi \ \ \text{and} \ \ \mathcal{S}_{13}^{B} = \frac{1}{NT} \sum_{t=1}^T \mathbf{1}^{\prime} \widetilde{ \mathbb{Y}}_{t-1}
\end{align}
\color{blue}
We have that $N^{-2} \mathbf{1}^{\prime} Q \mathbf{1} \to 0$ and $N^{-1} \sum_{j=0}^{\infty} \left\{ \mathbf{1}^{\prime} G^{j} \mathbf{1} \right\}^{1/2} \to 0$, as $N \to \infty$, which implies that $\mathcal{S}_{13}^{A} \to 0$ and $\mathcal{S}_{14}^{B} \to 0$. Therefore, $\mathcal{S}_{13} \to \frac{\beta_0}{1 - \beta_1 - \beta_2}$.   
\color{black}

\newpage

\subsection{Statistical Hypothesis Testing}

In this Section, we test hypotheses of interest based on the IVX-Wald test for the NVAR(1) model. For instance, we can examine whether a subset of parameters of the NVAR(1) model is zero. We begin, with a more simplified example since we are particularly  interested to check the presence of both predictability and causality of certain regressors. The first example, we consider is to test whether a subset of parameters of the NVAR(1) model are simultaneously zero across the nodes.  

To construct such a test we define $\mathcal{I} \subset \left\{ (j,r,s): j,r \in \left\{ 1,...,p  \right\} \ \text{and} \ s \in \left\{ 1,...,d \right\} \right\}$ be a subset of indices due to the presence of both nonstationary regressors and lagged $Y_{i,(t)}$ variables as predictors. We consider the following testing hypothesis: 
\begin{align*}
H_0&: \theta^{(s)}_{jr} = 0, \text{for all} \ (j,r,s) \in \mathcal{I}  \\
H_1&: \text{There exists at least one pair} \ (j,r,s) \in \mathcal{I} \ \text{such that} \ \theta^{(s)}_{jr} \neq 0. 
\end{align*}

Here we need to check whether the linear restrictions imposed on the parameter space of the NVAR(1) model under both the null and the alternative hypothesis do not violate the stability of the model. For example, whether we can establish that det$\left( \theta^{*} (z) \right) \neq 0$, for all $z \leq 1$. Furthermore, although the presence effects and nonlinearities in network dependence is an interesting research avenue (e.g., see \cite{armillotta2022nonlinear}, \cite{terasvirta1994aspects}), we focus on testing the simultaneous presence of nonstationarity and network dependence, which is a novel test on its own right in the existing literature, using our proposed estimation and inference methodology for NVAR models with nonstationary regressors.

\newpage 


\newpage


\section{Empirical Application}
\label{section4}

We present an empirical implementation based on the dataset of \cite{hardle2016tenet}. The particular dataset contains the stock returns of the top 100 US financial institutions by market capitalization with weekly time observations spanning the trading period from 2007 to 2013 as well as a set of macroeconomic and financial variables with unknown persistence properties. We begin by explaining the procedure to estimate the nonstochastic adjacency matrix for the full sampling period. 

In particular, we construct a tail network which is a suitable to capture tail interconnectedness (see, \cite{chen2019tail}), that is, financial conntectedness estimated with tail risk measures (see, also \cite{daouia2022extremile}). Suppose that $t \in \left\{ 1,...,T \right\}$, then by employing a quantile optimization function\footnote{Consider the conditional quantile function estimated via $Q_{y_i} \left( \tau | x_i \right) = F_{y_i}^{-1} \left( \tau | x_i \right)$. The optimization function to obtain the model estimates is expressed as $Q_{\tau} \left( y_i | x_i \right) = argmin_{q(x)} \ \mathbb{E} \big[ g_{\tau} \left( y_i - q(x_i) \right) \big]$, where $\tau \in (0,1)$ is a specific quantile level, and $g_{\tau}( u ) = u \left( \tau - \mathbf{1} { \left\{ u < 0 \right\} } \right)$ is the check function.} we can estimate the risk measures of the VaR and CoVaR via the following 
\begin{align}
R_{i,t} &= a_i  +  b_i^{\prime} X_{t-1} + u_{i,t} \label{VaR} \\
R_{j,t} &= a_{j|i} + b_{j|i}^{\prime} X_{t-1} + \gamma_{j|i} R_{i,t}  + u_{j|i,t} \label{CoVaR}
\end{align} 
where $\{ R_{i,t} \}_{i=1,...,N}$ is the vector of portfolio returns at time $t$, and $X_{t-1}$ is a vector of exogenous regressors containing macroeconomic characteristics common across assets. Let $Q_{\tau} \left( \cdot\ | \ \mathcal{F}_{t-1} \right)$ denote the quantile operator for $\tau \in (0,1)$ conditional on an information set $\mathcal{F}_{t-1}$. Then, the error terms satisfy $Q_{\tau} \left( u_{i,t} \ | \ X_{t-1} \right) = 0$ for the quantile regression \eqref{VaR}, and $Q_{\tau} \left( u_{j|i,t} \ | \ R_{it}, X_{t-1} \right) = 0$ for the quantile regression  \eqref{CoVaR}. 

Denote $\theta_{(1)} := \left( a_i, b_i^{\prime} \right)$ and $\check{X}_{t-1}^{(1)} := \left( 1, X_{t-1}^{\prime} \right)^{\prime}$ for \eqref{VaR} and $\theta_{(2)} := \left( a_{j|i}, b_{j|i}^{\prime}, \gamma_{j|i} \right)$ and  $\check{X}_{t-1}^{(2)} := \left( 1, X_{t-1}^{\prime}, R_{i,t} \right)^{\prime}$ for \eqref{CoVaR}. Then, the QR  estimator for model \eqref{VaR} is given by
\begin{align}
\label{opt1}
\widehat{ \theta}_{(1)} \left( \tau  \right) := \underset{ \theta_{(1)} \in \mathbb{R}^{p+2} }{  \text{argmin} } \sum_{t=1}^{T} g_{\tau} \left( R_{i,t} - \theta_{(1)}^{\prime} \check{X}_{t-1}^{(1)}  \right)
\end{align}
Similarly, the quantile estimator of \eqref{CoVaR} is obtained by
\begin{align}
\label{opt2}
\widehat{ \theta }_{(2)} \left( \tau  \right) := \underset{ \theta_{(2) } \in \mathbb{R}^{p+1} }{  \text{argmin} } \sum_{t=1}^{T} g_{\tau} \left( R_{j,t} - \theta_{(2)}^{\prime} \check{X}_{t-1}^{(2)}  \right)
\end{align}
The QR estimator from the optimization function \eqref{opt1} allows to construct the one-period ahead forecast for the Value-at-Risk of firm $i$ such that $\hat{ \text{VaR} }_{i,t+1} ( \tau ) = \hat{a}_i + \hat{ b }_i X_{t}$. Similarly, we can obtain the one-period ahead forecast for the Conditional-Value-at-Risk of firm $i$ such that $\hat{ \text{CoVaR} }_{j|i,t+1} ( \tau ) = \hat{a}_{j|i} + \hat{ b }_{j|i} X_{t} +  \hat{\gamma}_{j|i} \hat{VaR}_{i,t+1} ( \tau )$. This procedure allows to construct the VaR-$\Delta$CoVaR risk matrix proposed by \cite{katsourisOlmo20}. In this paper, we consider the construction of a binary adjacency matrix $\Omega_{ij}$, with elements such that $\omega_{ij} = 1$, if $\gamma_{i|j}$ is individually statistical significant and $\omega_{ij} = 0$ otherwise.

\newpage 

The above approach allows to obtain a non-random adjacency matrix which captures the tail connectivity of the nodes for the whole sampling period. The second step, is to estimate the NVAR(1) model assuming the regressors are generated by the local-to-unit root process as described in Section \ref{section2}. The testing hypothesis we are interested to examine is whether there is joint predictability and causality, that is, there exists some $\xi_i \neq 0$, for $i \in \left\{ 1,...,p \right\}$ which is a subset of the parameter space of $\xi$ and $\beta_2 \neq 0$. Under the null hypothesis there is no joint predictability and no causality while under the alternative hypothesis there is simultaneously joint predictability and causality.   


\newpage 

\section{Conclusion}
\label{section5}

In this paper, we consider a high-dimensional network and propose the Network Vector Autoregression model as a suitable statistical methodology to capture the network dynamics. The particular NVAR model we propose, allows to incorporate a nonstochastic adjacency matrix as well as it captures the time series properties of regressors via the local-to-unit root specification. Moreover, we examine the asymptotic theory of the model within the framework of network dependence and conditional neighborhood dependence presented by \cite{kojevnikov2021limit} and   \cite{lee2019stable} respectively. More specifically, we develop the asymptotic theory for the IVX estimator of the NVAR model and show that the estimator is robust to the degree of persistence of regressors, and specify the matrix moments of its limiting distribution.  To summarize, in this paper we consider the persistence properties of the regressors which correspond to each node in the network and develop an asymptotic theory based on the IVX instrumentation. An empirical application and two additional applications demonstrate the usefulness of our framework.

\bigskip

\paragraph{Conflicts of interest}

The author declares that there are no known conflicts of interest.

\paragraph{Data availability}

No data was used for the research described in this article.

\paragraph{Acknowledgements}

I wish to thank Professor Jose Olmo and Professor Tassos Magdalinos from the Department of Economics, University of Southampton for helpful discussions as well as Dr. Julius Vainora from the Faculty of Economics, University of Cambridge. Moreover, I am grateful to Professor Markku Lanne and Professor Mika Meitz from the Faculty of Social Sciences, University of Helsinki for helpful conversations. Financial support from the Research Council of Finland (grant 347986) is gratefully acknowledged.  All remaining errors are my own responsibility.

\newpage

\section{Appendix}



\subsection{Illustrative Example of VAR Representations}

\begin{example}

We consider the multivariate autoregressive index model representation proposed by \cite{reinsel1983some},  which provide a suitable parametrization for demensionality reduction. Specifically, consider an $n-$dimensional stationary vector autoregressive time series denoted by $\boldsymbol{Y}_t = ( y_{1t},..., y_{nt} )^{\prime}$ 
\begin{align}
\boldsymbol{Y}_t - \sum_{j=1}^p \boldsymbol{\Phi}_j \boldsymbol{Y}_{t-j} = \boldsymbol{\varepsilon}_t,    
\end{align}
where $\boldsymbol{\varepsilon}_t$ are \textit{i.i.d} $\mathcal{N} ( 0, \boldsymbol{\Omega} )$. Define with $\boldsymbol{\Phi} ( L ) = \boldsymbol{I} -  \boldsymbol{\Phi}_1 L - ... - \boldsymbol{\Phi}_p L^p$, where $L$ denotes the lag operator, and assume that $\mathsf{det} [ \boldsymbol{\Phi} (z) ] \neq 0$ for all complex numbers $| z | \leq 1$. In addition, we consider that $\boldsymbol{I} - \boldsymbol{\Phi} ( L ) = \sum_{j=1}^p \boldsymbol{\Phi}_j L^j$ can be factorized as $\boldsymbol{I} - \boldsymbol{\Phi} ( L ) \equiv \boldsymbol{A}( L ) \boldsymbol{B}(L)$ such that
\begin{align}
\boldsymbol{A}( L )_{( m \times r)}  &= \boldsymbol{A}_1 L + ... +  \boldsymbol{A}_{p_1} L^{p_1}
\\
\boldsymbol{B}( L )_{( r \times m)}  &= \boldsymbol{B}_0 + \boldsymbol{B}_1 L + ... +  \boldsymbol{B}_{p_2} L^{p_2}
\end{align}
Therefore, we have that 
\begin{align}
\boldsymbol{Y}_t =  \boldsymbol{A}( L ) \boldsymbol{B}(L) + \boldsymbol{\varepsilon}_t 
\equiv 
\sum_{ \mathsf{u} = 1 }^{p_1} \sum_{ \mathsf{v} = 1 }^{p_2} \boldsymbol{A}_{ \mathsf{u} } \boldsymbol{B}_{ \mathsf{v} } \boldsymbol{Y}_{ t - \mathsf{u} - \mathsf{v} } + \boldsymbol{\varepsilon}_t. 
\end{align}
where $p_1 + p_2 = p$. Suppose that $\boldsymbol{\eta}_t$ is an $r-$dimensional series such that $\boldsymbol{\eta}_t = \boldsymbol{B} (L) \boldsymbol{Y}_t \equiv \boldsymbol{B}_0 \boldsymbol{Y}_t + \boldsymbol{B}_1 \boldsymbol{Y}_{t - 1} + ... + \boldsymbol{B}_{ p_2 } \boldsymbol{Y}_{t - p_2}$ then the variables $\left\{ \boldsymbol{\eta}_{t-1},...,  \boldsymbol{\eta}_{t-p_1}  \right\}$, provide an adaptive filtration of all past information required for prediction. In particular, we have that 
\begin{align}
\boldsymbol{Y}_t =  \boldsymbol{A}( L ) \boldsymbol{\eta}_t   + \boldsymbol{\varepsilon}_t \equiv \sum_{ \mathsf{u} = 1 }^{p_1} \boldsymbol{A}_{ \mathsf{u} } \boldsymbol{\eta}_{ t - \mathsf{u}  } + \boldsymbol{\varepsilon}_t, 
\end{align}
Assume that $p_2 = 0$ and $p_1 = p$, then the model can be formulated as below
\begin{align}
 \boldsymbol{Y}_t = \sum_{j=1}^p \boldsymbol{A}_j \boldsymbol{B}_0 \boldsymbol{Y}_{t-j} + \boldsymbol{\varepsilon}_{t}  \equiv \boldsymbol{ \mathcal{A}} \boldsymbol{X}_{t-1} +  \boldsymbol{\varepsilon}_t
\end{align}
where $\boldsymbol{X}_{t-1}^{\prime}  = \big(  \boldsymbol{Y}_{t-1}^{\prime}  \boldsymbol{B}_0^{\prime},..., \boldsymbol{Y}_{t-p}^{\prime}  \boldsymbol{B}_p^{\prime}  \big) = \big( \boldsymbol{Y}_{t-1}^{\prime},..., \boldsymbol{Y}_{t-p}^{\prime}  \big) \big( \boldsymbol{I}_p \otimes \boldsymbol{B}_0^{\prime} \big)$ and $\boldsymbol{ \mathcal{A}} = ( \boldsymbol{A}_1,...., \boldsymbol{A}_p )$. 

In addition, full rank $r$ conditions for these matrices are assumed to hold. Since the elements of $\boldsymbol{A}_j^{\prime}$ and $\boldsymbol{B}_0$ are determined only up to nonsingular linear transformations, that is, $ \boldsymbol{A}_j \boldsymbol{B}_0 \equiv \boldsymbol{A}_j \boldsymbol{P}^{-1} \boldsymbol{P} \boldsymbol{B}_0$, for $j = 1,..., p$ and any $( r \times r)$ nonsingular matrix $\boldsymbol{P}$, we must impose some normalization conditions to ensure uniqueness of the parameters. This example can be extended to more complex functional forms provided that we can obtain observationally equivalent processes. 
\end{example}

\begin{example}[see, \cite{cubadda2022dimension}]
Consider the Dimension-Reducible VAR model as
\begin{align}
\boldsymbol{Y}_t = \sum_{j=1}^p \boldsymbol{A} \alpha_j \boldsymbol{A}^{\prime} \boldsymbol{Y}_{t-j} + \boldsymbol{u}_t,     
\end{align}
Under the assumption that joint DGP of the observed variables $\boldsymbol{Y}_t$ follows that of a FAVAR then the coefficient matrix has the following structure
\begin{align}
\boldsymbol{A} = 
\begin{bmatrix}
\boldsymbol{I}_m & \boldsymbol{0}_{  (n-m) \times m }
\\
\boldsymbol{0}_{ m \times  (n-m) } & \boldsymbol{B}_{ (n-m) \times (r -m ) }
\end{bmatrix}
\end{align}
Therefore, to perform structural analysis through the DRVAR, then the invertability condition for the polynomial VAR coefficient matrix should hold to obtain the Wold representation of the time series $\boldsymbol{Y}_t$. However, we can invert the polynomial coefficient matrix of $x_t = \sum_{j=1}^p \alpha_j x_{t-j} + \xi_t$ and insert the Wold representation of the dynamic components $x_t$ in expression $Y_t = A x_t + \varepsilon_t$ to obtain  
\begin{align}
\boldsymbol{Y}_t = \boldsymbol{A} \boldsymbol{\gamma} (L) \boldsymbol{\xi}_t + \boldsymbol{\varepsilon}_t,     
\end{align}
where $\boldsymbol{\gamma}(L)^{-1} = \boldsymbol{I}_n - \sum_{j=1}^p \alpha_j L^j$. Lastly, by linearly projecting $\varepsilon_t$ on $\xi_t$, we can decompose the static component as $\varepsilon_t = \rho \xi_t + v_t$, where $\rho = A_{\perp} A_{ \perp }^{\prime} \Sigma_u \boldsymbol{A}  ( \boldsymbol{A}^{\prime} \Sigma_u \boldsymbol{A} )^{-1}$ 
\begin{align}
\boldsymbol{Y}_t = \underbrace{ \boldsymbol{C} (L) \boldsymbol{\xi}_t }_{ \textcolor{blue}{ \chi_t } } + \boldsymbol{v}_t,     
\end{align}
where $C_0 = (A + \rho)$ and $C_j = A \gamma_j$ for some $j > 0$.  

The above derivations show how to decompose the underline dynamics of the observable series $\boldsymbol{Y}_t$ into the common component $\boldsymbol{\chi}_t$, and the ignorable errors, $\boldsymbol{v}_t$. Moreover, since the errors $\boldsymbol{\xi}_t$ and $\boldsymbol{v}_t$ are uncorrelated at any lead and lags, we can recover the structural shocks solely by the reduced form errors $\boldsymbol{\xi}_t$ of the common component $\boldsymbol{\chi}_t$ using procedures commonly used in structural VAR analysis. In particular, we can obtain the structural shocks as $\boldsymbol{u}_t = \boldsymbol{C}^{-1} \boldsymbol{D} \boldsymbol{\xi}_t$ and the impulse response functions from $\boldsymbol{\Psi} (L) = \boldsymbol{C}(L) \boldsymbol{D}^{-1} \boldsymbol{C}$, where $\boldsymbol{D}$ is the matrix formed by the first $r$ rows of $\boldsymbol{C}_0$ and $\boldsymbol{C}$ is the lower triangular matrix such that $\boldsymbol{C} \boldsymbol{C}^{\prime} = \boldsymbol{D} \boldsymbol{A}^{\prime} \boldsymbol{\Sigma}_u \boldsymbol{A} \boldsymbol{D}^{\prime}$. Notice that such identification strategy is based on a unique rotation of the reduced form common shocks $\boldsymbol{\xi}_t$, and hence it does not require to endow the dynamic component $\boldsymbol{\chi}_t$ with an economic interpretation.
\end{example}

\newpage

\subsection{Illustrative Examples}

\begin{example}[NVAR$(p,q)$]
The main idea of a NVAR-LUR$(p,q)$ model is as follows. To begin with, innovation transmission across bilateral links takes place with a lag. Moreover, uncertainty occurs due to unequal variation in the frequency of network interactions in comparison to the time series observations. Thus, the proposed econometric functional form specification accommodates characteristics on how innovations transmit through the graph (network) over time. 
\begin{align}
\label{NVAR}
x_t = \alpha_1 A x_{t-1} + ... + \alpha_p A x_{t-p} + v_t \equiv A \sum_{j=1}^p \theta_j x_{t-j} + v_t,  
\end{align}
\end{example}

\begin{proposition}[Granger-Causality in NVAR$(p,1)$]
\

Suppose that $x_t$ is generated by \eqref{NVAR} and assume that $\alpha _{ \ell } \neq 0 \ \forall \ \ell \in \left\{ 1,..., p \right\}$, $x_j$ Granger-causes $x_i$ at horizon $h$ \textit{iff} there exists a connection from $i$ to $j$ of at least one order $k \in \left\{ k^{\star}, k^{\star} + 1,..., h \right\}$, where $k^{\star} = \mathsf{ceil} ( h / p )$.       
\end{proposition}

Consequently, the GIRF has the following form: 
\begin{align}
\frac{ \partial x_{i,t + h} }{ \partial v_{j,t} } \bigg| \mathcal{F}_t 
:= 
\vartheta_{t = k^{\star} } ( \alpha ) [ A_{ij} ] |_{ t = k^{\star} } + ... +  \vartheta_{t = k^{\star} } ( \alpha ) [ A_{ij} ] |_{ t = h } 
\end{align}
The coefficients $\left\{  \vartheta_{t} ( \alpha )  \right\}_{ t = k^{\star} }^h$ are polynomials of $\left\{ \alpha_j \right\}_{ j = 1}^p$.  Our aim is to bridge the interaction between network connectedness and transmission of shock dynamics under the presence of possibly nonstationarity.

\medskip

\begin{proof}
Consider again the functional form of the NVAR$( p, 1 )$ model as below: 
\begin{align}
y_t = \alpha_1 A y_{t-1} + ... + \alpha_p A y_{t-p} + u_t, 
\end{align}    
Let $\boldsymbol{\alpha} = ( \alpha_1,..., \alpha_p )^{\prime} \in \mathbb{R}^p$ and denote with $\boldsymbol{X}_t$ the covariates matrix which includes information in the lags 1 to $p$ of $y_t$ using first-order network connections,
\begin{align}
\boldsymbol{X}_t = \big[ A y_{t-1},...,  A y_{t-p} \big]    
\end{align}
Since the network adjacency matrix is assumed to be fixed (time invariance in social interactions), then this implies that the dependence of the regressors on the network dependence can suppressed. This implies that given a known adjacency matrix $A$, then the unknown parameter vector $\boldsymbol{\alpha}$ can be estimated by OLS, especially due to the absence of exogenous regressors or nonstationary regressors in the system.

\newpage 

This yields the following optimization problem: 
\begin{align}
\underset{ \boldsymbol{\alpha} \in \mathbb{R}^{p+1} }{ \mathsf{min} } \ \frac{1}{NT} \sum_{t=1}^T \big( y_t - \boldsymbol{X}_t \boldsymbol{\alpha} \big)^{\prime} \boldsymbol{\Sigma} \big( y_t - \boldsymbol{X}_t \boldsymbol{\alpha} \big),
\end{align}
where $\boldsymbol{\Sigma} = \mathsf{Var} ( u_t )$. Therefore, this implies that 
\begin{align}
\hat{\boldsymbol{\alpha}}_{ols}  = \left( \sum_{t=1}^T \boldsymbol{X}_t^{\prime} \boldsymbol{\Sigma}^{-1} \boldsymbol{X}_t \right)^{-1} \left( \sum_{t=1}^T \boldsymbol{X}_t^{\prime} \boldsymbol{\Sigma}^{-1} y_t \right)  
\end{align}
Specifically, when we assume that the covariance matrix of the disturbances is the identity matrix, that is, $\boldsymbol{\Sigma} = \boldsymbol{I}$, then the OLS estimator takes the form of a pooled OLS estimator as below: 
\begin{align}
\hat{\boldsymbol{\alpha}}_{ols} =  \left( \sum_{t=1}^N \sum_{t=1}^T x_{it} x_{it}^{\prime} \right)^{-1} \left( \sum_{t=1}^N \sum_{t=1}^T x_{it} y_t \right)    
\end{align}
As $n \to \infty$, then it holds that 
\begin{align}
\sqrt{N} \left(  \hat{\boldsymbol{\alpha}}_{ols} - \boldsymbol{\alpha} \right) \Rightarrow \mathcal{N} \left(  0, \frac{\sigma^2}{ T } \mathbb{E} \left[ x_{it} x_{it}^{\prime} \right] \right)   
\end{align}
On the other hand, as $T \to \infty$ then the $\hat{\boldsymbol{\alpha}}_{ols}$ estimator is consistent if the model is specified correctly and $y_t$ is ergodic and strictly stationary which implies that
\begin{align}
\sqrt{T} \left(  \hat{\boldsymbol{\alpha}}_{ols} - \boldsymbol{\alpha} \right) \Rightarrow \mathcal{N} \big(  0, \mathbb{E} \left[ \boldsymbol{X}_{t} \boldsymbol{X}_{t}^{\prime} \right]^{-1} \mathbb{E} \left[ \boldsymbol{X}_{t}^{\prime} \boldsymbol{\Sigma}  \boldsymbol{X}_{t}^{\prime} \right]   \mathbb{E} \left[ \boldsymbol{X}_{t} \boldsymbol{X}_{t}^{\prime} \right]^{-1 \prime} \big)   
\end{align}

\medskip

\begin{remark}
Now, there are certain cases in which although the estimator $\boldsymbol{\alpha}$ that corresponds to the NVAR$(p,1)$ model can be obtained via the Expectation-Maximization (EM) algorithm, point identification of the estimator is not guaranteed. This might suggest that the mapping between parameters in the process for $\left\{ y_t \right\}_{ t=1 }^T$ and $\alpha$ is not bijective, similar to the statistical problem of estimating continuous time models using discrete time data. Furthermore, Bayesian methods can be employed especially in cases of lack of point identification by imposing specific distributional assumptions on the Prior and Posterior distributions corresponding to the solution of the EM algorithm for the model estimator. 
\end{remark}

Notice that the process $\left\{ y_t \right\}$ is weakly stationary \textit{iff} for all eigenvalues $\lambda_i$ of $A$ it holds that $| \lambda_i | < 1 / | a |$. 

\end{proof}

\newpage 

\subsection{Impulse Response Functions}

Assume that $y_t$ is stationary, then the long-term response of $y_t$ to a permanent increase in $u_t$ is equivalent to the (contemporaneous) response of $y$ to a disturbance in $\epsilon$, $\partial y / \partial \epsilon$, such that
\begin{align}
R = \underset{ h \to \infty  }{ \mathsf{lim} } \ \left(  \frac{ \partial y_{t+h} }{ \partial u_t  } + \frac{ \partial y_{t+h} }{ \partial u_{t+1}  } + .... +   \frac{ \partial y_{t+h} }{ \partial u_{t+h}  } \right)  \equiv \frac{ \partial  }{ \partial \varepsilon  }
\end{align}
Notice that it holds that $y = \big( I - \alpha A \big)^{ -1 } \varepsilon$. Moreover, since $y_t$ is assumed to be a stationary process then it holds that 
\begin{align}
R 
= 
\underset{ h \to \infty  }{ \mathsf{lim} } \ \sum_{ j = 0 }^{ h + 1 } \frac{ \partial y_{t+h} }{ \partial u_{ t + h - j } } 
=  
\underset{ h \to \infty  }{ \mathsf{lim} } \ \sum_{ j = 0 }^{ h + 1 } \frac{ \partial y_{t+h} }{ \partial u_{ t  } }   =  \sum_{ j = 0 }^{ \infty }  \frac{ \partial y_{t+h} }{ \partial u_{ t  } }.  
\end{align}
Therefore, to obtain the impulse response for $x_t$, one can write the econometric specification in the following companion form: 
\begin{align}
\boldsymbol{z}_t = \boldsymbol{F} \boldsymbol{z}_{t-1} + \boldsymbol{e}_t,     
\end{align}
where the $( N \times N )$ $\boldsymbol{F}$ matrix is expressed as below:
\begin{align}
\boldsymbol{F} =
\begin{bmatrix}
\alpha_1 A \ & \ \alpha_2 A \ & \hdots \ & \ \alpha_{p-1} A \ & \ \alpha_p A 
\\
I_N \ & \ 0_N \ & \ \hdots \  & \ 0_N \ & \ 0_N \ 
\\
0_N & I_N  & \hdots &  0_N  &  0_N
\\
\vdots \ & \ \vdots \  & \  \ddots  \  & \ \vdots \ & \ \vdots
\\
0_N \ & \ 0_N \ & \ \hdots \ & \ I_N \ & 0_N 
\end{bmatrix}
\end{align}
Therefore, the impulse response of $x_t$ to a disturbance in $v_t$ is then given by $( n \times n )$ upper left block in $F^h$ such that:
\begin{align}
\frac{ \partial y_{t+h} }{ \partial u_t } 
=  
\frac{ \partial y_{t+h} }{ \partial z_{t+h} } \cdot  \frac{ \partial z_{t+h} }{ \partial e_t } \cdot \frac{ \partial e_t }{ \partial u_t }   
\end{align}

\newpage

\begin{center}
\textbf{SUPPLEMENTARY APPENDIX} 
\\ 
\vspace{0.5em}
"\textbf{Robust Estimation in Network Vector Autoregression with Nonstationary Regressors}"    
\\
\vspace{1em}
Christis Katsouris 
\\
\textit{University of Southampton} $\&$ \textit{University of Helsinki}
\end{center}

\bigskip

\appendix

\section{Supplementary Results}
\label{Section9}

Let $\mathcal{D}$ be the space of c\`adl\`ag functions $f : [0,1] \to \mathbb{R}$ equipped with the Skorokhod topology. As a measurable structure on $\mathcal{D}$ we consider the corresponding Borel $\sigma-$algebra $\mathcal{B}( D )$.  We write $X_n \overset{ \mathcal{D} }{\to} X$ whenever $X_n$, $X$ are random variables taking values in $\big( \mathcal{D}, \mathcal{B} ( \mathcal{D} ) \big)$ such that $X_n$ converges weakly to $X$ in $\mathcal{D}$ as $n \to \infty$, that is, $\underset{ n \to \infty }{ \mathsf{lim} } \mathbb{E} f (X_n ) \to \mathbb{E} f(X)$ for all bounded continuous functions $f : \mathcal{D} \to \mathbb{R}$.     
\color{black}

\subsection{Matrix Representation of multivariate LUR process}

Consider that the $p-$dimensional vector of the system's predictors $\boldsymbol{X}_t$ is generated via the following multivariate LUR process (see, also \cite{katsouris2023statistical})
\begin{align}
\label{LUR}
\boldsymbol{X}_t = \left( \boldsymbol{I}_p - \frac{ \boldsymbol{C}_p }{ T^{ \upgamma_x } } \right) \boldsymbol{X}_{t-1} + \boldsymbol{u}_t, \ \ \ \ t \in \left\{ 1,...,T \right\} 
\end{align}
Let $\boldsymbol{\Phi}_T ( c_i, \upgamma_x ) := \left( \boldsymbol{I}_p - \frac{ \boldsymbol{C}_p }{ T^{ \upgamma_x } } \right)$, then under the assumption of homogeneous coefficients of persistence, that is,  $c_i \equiv c$ for all $i \in \left\{ 1,..., p \right\}$, then \eqref{LUR}
can be written as below
\begin{align}
\label{le.repr}
\boldsymbol{\Phi}_T ( c_i, \upgamma_x )  \boldsymbol{X} = \boldsymbol{U}
\end{align}
where the non-stochastic $T \times T$ matrix $\boldsymbol{\Phi}_T ( c_i, \upgamma_x )$ takes the bi-diagonal form as below
\begin{align}
\boldsymbol{\Phi}_T ( c_i, \upgamma_x ) = 
\begin{bmatrix}
\textcolor{red}{1} & 0 & 0 & \hdots & 0  & 0 
\\
\textcolor{blue}{- \varphi_T ( c, \upgamma_x )} & \textcolor{red}{1} & 0 & \hdots & 0  & 0 
\\
0 & \textcolor{blue}{- \varphi_T ( c, \upgamma_x )} & \textcolor{red}{1} & \vdots & 0  & 0 
\\
\vdots & \vdots & \vdots & \vdots & \vdots  & \vdots
\\
0 & 0 & 0 & 0 & \textcolor{blue}{- \varphi_T ( c, \upgamma_x )}  & \textcolor{red}{1} 
\end{bmatrix}
\end{align}
with $\textcolor{blue}{- \varphi_T (c_i, \upgamma_x ) = ( 1 -  \frac{c_i}{T^{ \upgamma_x } } )}$. 

\newpage

Denote with $\boldsymbol{K}_T = \boldsymbol{I}_T - \frac{1}{T} \mathbf{1} \mathbf{1}^{\top}$ denote the $T \times T$ centering matrix and $\boldsymbol{X}_K = \boldsymbol{K}_T \boldsymbol{X}$ the corresponding centred big data matrix. Therefore, the corresponding centered covariance matrix can be written in the following form 
\begin{align}
\label{le.covariance}
\boldsymbol{ \mathcal{S} }_T = \frac{1}{p} \boldsymbol{X}_K \boldsymbol{X}_K^{\top}. 
\end{align} 

\begin{remark}
\

\begin{itemize}
\item Firstly, notice that the centering matrix $\boldsymbol{K}_T$ is employed for demeaning the original data in high dimensional settings. One application can be found in the representation of the dynamic time series model (see, Handbook of Econometrics). A different implementation of this property is employed by \cite{maillet2015global} as well as by  \cite{davis2016asymptotic} who study the asymptotic theory for the sample covariance matrix of a heavy-tailed multivariate matrix. Another application is presented by \cite{Phillips2013inconsistent} who derive an expression of the condition number in explosive VAR models.  

\item Secondly, notice that the covariance matrix representation given by expression \eqref{le.covariance} becomes a functional of the nuisance coefficient of persistence, which is a data specific feature. However considering the structure of our proposed novel VaR-CoVaR risk matrix, the asymptotic theory of its largest eigenvalues will be more challenging than conventional approaches. In this case, we will need to consider both the framework proposed by \cite{davis2016asymptotic} as well as the proposed regression-based estimation methodology for the risk matrix, the quantile predictive regressions which capture the possibly nonstationary nature of predictors; in order to derive an analytical tractable limiting distribution with desirable properties for statistical inference. 

\end{itemize}
\end{remark}
Moreover, we have that $\boldsymbol{X} = \boldsymbol{\Phi}_T ( c_i, \upgamma_x )^{-1} \boldsymbol{U}$ which implies that the covariance type matrix $\boldsymbol{ \mathcal{S} }_T$ can be written in the following form
\begin{align}
\boldsymbol{ \mathcal{S} }_T = \frac{1}{p} \left( \boldsymbol{K}_T \boldsymbol{\Phi}_T ( c_i, \upgamma_x )^{-1} \boldsymbol{U} \right) \left( \boldsymbol{K}_T \boldsymbol{\Phi}_T ( c_i, \upgamma_x )^{-1} \boldsymbol{U} \right)^{\top}  
\end{align}
Then, the uncentered matrix can be written as below 
\begin{align}
\boldsymbol{ \mathcal{S} }^o_T  = \frac{1}{p} \boldsymbol{X} \boldsymbol{X}^{ \top }.
\end{align}
and using the multivariate LUR process representation \eqref{le.repr} we obtain
\begin{align}
\boldsymbol{ \mathcal{S} }^o_T = \left( \boldsymbol{\Phi}_T ( c_i, \upgamma_x )^{-1} \boldsymbol{U} \right) \left( \boldsymbol{\Phi}_T ( c_i, \upgamma_x )^{-1} \boldsymbol{U} \right)^{\top}
\end{align}
Denote with $\hat{ \lambda}_k$ and $\hat{ \lambda}^{o}_k$ the $k-$th largest eigenvalues of $\boldsymbol{ \mathcal{S} }_T$ and $\boldsymbol{ \mathcal{S} }^o_T$ respectively. Then, one would be interested to examine the asymptotic behaviour of the largest eigenvalues of the matrices $\boldsymbol{ \mathcal{S} }_T$ and $\boldsymbol{ \mathcal{S} }^o_T$ in a high-dimensional setting.

\newpage

\newpage

\subsection{Predictive Regression with variable addition}

Consider the model 
\begin{align}
y_t &= \beta x_{t-1} + u_t \\
x_t &= \rho  x_{t-1} + v_t
\end{align}
with the IVX instrument defined as 
\begin{align}
z_{t-1} = ( 1 - \rho_z L )^{-1} \Delta x_{t-1}, \ \ \rho_z = \left( 1 - \frac{c_z}{ T^{ \lambda } } \right)
\end{align}
For the IVX estimator it holds that $\hat{\beta}_{ivx} - \beta = \mathcal{O}_p \left(  T^{- \frac{ \lambda + 1 }{2} }  \right)$. Moreover, the limiting distribution of the IVX-based t type statistic $t_{ivx}$ for the null $\beta = 0$, under the sequences of local alternatives of the type $\beta = \beta_T = b T^{ - \frac{ \lambda + 1 }{2} }$, is shown to be
\begin{align}
t_{ivx} \Rightarrow \mathcal{Z} + b \frac{ \sigma_v \sqrt{2} }{ \sigma_u \sqrt{ c_z } } \left( B_c^2(1) - \int_0^1 B_c(s) dB_c(s) \right)
\end{align}
where $dB_c(s) := - c B_c(s) ds + dV(s)$ and $\mathcal{Z}$ is a standard normal variate independent of $V(t)$ and $B_c(s)$. 

\begin{proposition}
Denote by $\hat{\beta}_{ivx}^{*}$ the IVX estimator of $\beta$ with $y_{t-1}$ instrumented by itself and by $t_{ivx}^{*}$ the corresponding t statistic. Let $\sigma_{uv} \neq 0$, then we have that 
\end{proposition}

Under the assumptions of the proposition the following hold
\begin{align}
\frac{1}{ T^{ \frac{1+\lambda}{2}} } \sum_{t=2}^T z_{t-1} u_t &\Rightarrow \mathcal{N} \left( 0, \frac{\sigma_u^2 \sigma_v^2 }{2 c_z }  \right)
\\
\frac{1}{ T^{ 1+\lambda } } \sum_{t=2}^T z^2_{t-1} &\to \frac{ \sigma_v^2 }{ 2 c_z} 
\\
\frac{1}{ T^{ \frac{1+\lambda}{2}} } \sum_{t=2}^T z_{t-1} x_{t-1}  &\Rightarrow \frac{ \sigma_v^2 }{ c_z } \left( B_c^2(1) - \int_0^1 B_c(s) dB_c(s) \right), 
\end{align}

\newpage

\subsection{Limit Theory for Moderately Explosive Systems}

\paragraph{Case II: C with not distinct diagonal elements, $c_i = c_j$ for some $i \neq j$}

We now consider case (II), where the matrix $C$ does not have distinct diagonal elements. In this case two or more elements of $x_t$ have comparable moderately explosive behaviour governed by a common autoregressive root of the form $\rho_{nj} = 1 + c_j / n^{\alpha}$. We know that under such conditions the second moment matrix $\sum_{t=1}^n x_t x_t^{\prime}$ in the regression model is asymptotically singular, which is explained by the fact that some elements of $x_t$ have common moderately explosive behaviour. These elements of $x_t$ are then asymptotically multicollinear in much the same way as regressors that are cointegrated or have common deterministic trends. To deal with this singularity in the regression model, we can develop an asymptotic theory for the regression in a similar way by rotating the regression coordinates in the direction of the common explosive behaviour and in an orthogonal direction. Here, however, the rotation is a random process determined by the regressor vector $x_n$. The randomness and the sample size dependence in the rotation present further complications in the development of the asymptotics because the limit theory segementation then depends on weak convergence of the rotation matrix. In what follows, we develop the related asymptotic theory for the special case where $C$ is the scalar matrix $C = c I_K$ and $\rho_N = 1 + c / n^{\alpha}$.  

For $\left( k_n \right)_{n \in \mathbb{N} }$ any sequence increasing to infinity, we define
\begin{align}
Y_{ C_n } := \frac{1}{ n^{\alpha / 2} } \sum_{j=1}^{ k_n } R_n^{ -j } F_x(1) \epsilon_j. 
\end{align}

The stochastic sequence of $Y_{ C_n }$ plays an important role in determining the asymptotic behaviour of the explosive systems. 

\paragraph{Proof of Lemma 4.2 MP (2009)}

\begin{lemma}
For each $\alpha \in (0,1)$, $C > 0$ and a sequence $k_n$ satisfying $\norm{ R_n }^{- k_n } \to 0$ and $n^{\alpha} \norm{ R_n }^{- (n-k_n)} \to 0$, we have that, as $n \to \infty$, 
\begin{itemize}
\item[(i)]
\begin{align}
n^{- \alpha } \mathbb{E} \norm{ \text{vec} \sum_{t=1}^n u_{0t} \left(  \sum_{j=t+1}^n R_n^{t-j} u_{xj} \right)^{\prime} R_n^{-n} } \to 0,
\end{align} 
\item[(ii)]
\begin{align}
n^{- \alpha } \mathbb{E} \norm{ \text{vec} \sum_{t=1}^{k_n} u_{0t} \left(  \sum_{j=t+1}^n R_n^{t-j} u_{xj} \right)^{\prime} R_n^{-n} } \to 0.
\end{align} 
\end{itemize}
\end{lemma}

\newpage 

\begin{proof}
Notice that the sample covariance can be written as below
\begin{align*}
\text{vec} \frac{1}{n^{\alpha} } \sum_{t=1}^n u_{0t} x_t^{\prime} R_n^{-n} &=  
\frac{1}{n^{\alpha} } \text{vec} \left\{ \sum_{t=1}^n u_{0t} \left( \sum_{j=1}^t R_n^{t-j} u_{xj} \right)^{\prime} R_n^{-n} \right\} + o_p(1)
\\
&= 
\frac{1}{n^{\alpha} } \text{vec} \left\{ \sum_{t=k_n + 1}^n u_{0t} \left( \sum_{j=1}^n R_n^{t-j} u_{xj} \right)^{\prime} R_n^{-n} \right\} + o_p(1)
\\
&= 
\frac{1}{n^{\alpha / 2} } \sum_{t=k_n + 1}^n \left( R_n^{-(n-t)} \otimes u_{0t} \right) \left( \frac{1}{n^{\alpha / 2} } \sum_{j=1}^n R_n^{t-j} u_{xj} \right)^{\prime} R_n^{-n} + o_p(1)
\\
&= \big[ I_{mK} + o_p(1) \big] \frac{1}{n^{\alpha / 2} } \sum_{t= k_n + 1}^n \left( R_n^{-(n-t)} \otimes u_{0t} \right) Y_{C_n} + o_p(1).  
\end{align*}
Applying the BN decomposition on the proceeding expression and making use of  
\begin{align}
\frac{1}{n^{\alpha / 2} } \sum_{t=k_n + 1}^n \left( R_n^{-(n-t)} \otimes \Delta \tilde{ \epsilon }_{0t}  \right)  = o_p(1), \ \ \text{as} \ \ n \to \infty, 
\end{align}
we obtain the following expression for the sample covariance as $n \to \infty$
\begin{align*}
\text{vec} &\left\{ \frac{1}{n^{\alpha} } \sum_{t=1}^n u_{0t} x_t^{\prime} R_n^{-n} \right\}
\\
&= 
\big[ I_{mK} + o_p(1) \big] \frac{1}{n^{\alpha / 2} } \sum_{t=1}^{n - k_n} \left[  R_n^{-(n-k_n - t)} Y_{C_n} \otimes F_0(1) \epsilon_{t + \kappa_n }   \right].
\end{align*}
Furthermore, letting $\mathcal{F}_{n,i} := \sigma \left( x_0, \epsilon_i,  \epsilon_{i-1},... \right)$ since $Y_{C_n}$ is $\mathcal{F}_{n,k_n}$ measurable, then 
\begin{align}
\xi_{n, t+ k+n} := \frac{1}{n^{\alpha / 2} } \left[  R_n^{-(n-k_n - t)} Y_{C_n} \otimes F_0(1) \epsilon_{t + \kappa_n }   \right]. 
\end{align}
is an $\mathbb{R}^{mK}-$valued martingale difference array with respect to $\mathcal{F}_{n,t + k_n}$. Moreover, we denote by $M_{n,k} := \sum_{t=1}^k \xi_{n,t+k_n}$ the martingale array corresponding to $\xi_{n,t+k_n}$ and by $\langle M_{n,k} \rangle := \sum_{t=1}^k \mathbb{E}_{ \mathcal{F}_{n,t + k_n} } \left[ \xi_{n,t+k_n} \xi_{n,t+k_n}^{\prime} \right]$ the predictable quadratic variation of $M_{n,k}$. Since we have that
\begin{align*}
\langle M_{n,k} \rangle_{n - k_n} 
&= \left[ \frac{1}{n^{\alpha} } \sum_{t=1}^{ n - k_n }  R_n^{-(n-k_n - t)} Y_{C_n} Y_{C_n}^{\prime} R_n^{-(n-k_n - t)} \right] \otimes \Omega_{00}
\\
&= 
\big[ I_{mK} + o_p(1) \big] \left( \int_0^{ \infty } e^{ -pC } Y_{C_n}  Y_{C_n}^{\prime} e^{ -pC } \otimes \Omega_{00} \right)
\end{align*}
and $\langle M_{n,k} \rangle_{n - k_n}$ is $\mathcal{F}_{n, k_n}-$measurable with $\mathcal{F}_{n, k_n} \subset \mathcal{F}_{n, k_n + 1}$, that is, the $\sigma-$algebra supporting the predictable quadratic variation is smaller than each of the elements of the filtration supporting the array $\xi_{n, t+ k_n}$. 

\newpage 

Therefore, we obtain that 
\begin{align*}
\text{ vec } \left\{ \frac{1}{n^{\alpha} } \sum_{t=1}^{ n } u_{0t} x_t^{\prime} R_n^{-n} \right\}
&= M_{n, n - k_n} + o_p(1)
\\
&\Rightarrow \mathcal{MN} \left( 0, \int_0^{\infty} e^{ pC} Y_C Y_C^{\prime} e^{ pC} \otimes \Omega_{00} \right)
\end{align*} 
as $n \to \infty$. The verification of the Lindeberg and tightness conditions needed for this result is provided by Proposition A1 of MP.

Therefore, we can now deduce the limit behavior of the regression coefficient for the case (I), where the localizing coefficients are distinct. This is given by
\begin{align*}
\text{vec} \left\{ n^{\alpha} \left( \hat{A}_n - A \right) R_n^n \right\}
&= 
\left[ \left( \frac{1}{n^{2 \alpha }} \sum_{t=1}^n R_n^{-n} x_t x_t^{\prime} \right)^{-1} \otimes I_m \right] \text{vec} \left(  \frac{1}{n^{\alpha}}  \sum_{t=1}^n u_{0t} x_t^{\prime} R_n^{-n} \right)
\\
&=
\big[ I_{mK} + o_p(1) \big] \left[ \left( \int_0^{ \infty } e^{ -pC } Y_{C_n}  Y_{C_n}^{\prime} e^{ -pC } dp \right)^{-1} \otimes I_{m} \right] \sum_{t=1}^{n - k_n} \xi_{n,t} 
\\
&= \big[ I_{mK} + o_p(1) \big] \left[ \left( \int_0^{ \infty } e^{ -pC } Y_{C_n}  Y_{C_n}^{\prime} e^{ -pC } dp \right)^{-1} \otimes \Omega_{00}^{-1} \right] \left( I_{mK} \otimes \Omega_{00} \right) \sum_{t=1}^{n - k_n} \xi_{n,t} 
\\
&= \big[ I_{mK} + o_p(1) \big] \langle M_n \rangle_{ n - k_n}^{-1} \left( I_{mK} \otimes \Omega_{00} \right) M_{n, n - k_n}. 
\end{align*}

The limiting distribution of $M_{n, n - k_n}$ is established in (26). We also show in the Appendix that $M_{n, n - k_n}$ satisfies the requirements of Proposition A1 (iii), so that joint convergence of $M_{n, n - k_n}$  and $\langle M \rangle_{n - k_n}$ applies. Therefore, with the next theorem we prove that the regression coefficient has a mixed normal limiting distribution.
\end{proof}  

\begin{theorem}
For the model (1)-(2) with $R_n = I_K + C / n^{\alpha}$, $c_i > 0$ for all $i, c_i \neq c_j$ for all $i \neq j, \alpha \in (0,1)$, and weakly dependent errors satisfying Assumption LP, we have that
\begin{align}
n^{\alpha} \left( \hat{A}_n - A \right) R_n^n \Rightarrow \mathcal{MN} \left( 0, \left( \int_0^{ \infty } e^{ -pC } Y_{C_n}  Y_{C_n}^{\prime} e^{ -pC } dp \right)^{-1} \otimes \Omega_{00} \right)
\end{align}
\end{theorem}

\begin{proof}
Start by defining the orthogonal random matrix $H_n = [ H_{cn} , H_{ \perp n} ]$, where 
\begin{align}
H_{cn} = \frac{ x_n }{ \left( x_n^{\prime} x_n \right)^{1/2} }, \ \ H_{ \perp n}^{\prime} H_{cn} = 0 \ \ \text{almost surely}, 
\end{align} 
and the $K \times ( K - 1 )$ random matrix $H_{ \perp n}$ is an orthogonal complement to $H_{cn}$ satisfying $H_{ \perp n}^{\prime} H_{ \perp n} = I_{K-1}$ and $H_{ \perp n}^{\prime} H_{ \perp n} = I_{K} - H_{ cn}^{\prime} H_{cn}$ \textit{almost surely}. Therefore, the asymptotic behaviour of 
$H_{ \perp n} H_{ \perp n}^{\prime}$ is given by 
\begin{align}
H_{ \perp n} H_{ \perp n}^{\prime} = I_K - \frac{ Y_{cn} Y_{cn}^{\prime} }{   Y_{cn}^{\prime} Y_{cn} } + o_p(1) \Rightarrow I_K - \frac{ Y_{c} Y_{c}^{\prime} }{ Y_{c}^{\prime} Y_{c} } \ \ \text{as} \ n \to \infty, 
\end{align} 

\newpage

where $Y_{cn}$ and $Y_{c}$ are random vectors $Y_{C_n}$ and $Y_{C}$ of Lemma 4.1 with $C = c I_K$. Then, applying this rotation to the moderately explosive regressor vector yields
\begin{align}
z_t = H_n^{\prime} x_t 
=
\begin{bmatrix}
H_{cn}^{\prime} x_t \\
H_{\perp n}^{\prime} x_t
\end{bmatrix} 
:= 
\begin{bmatrix}
z_{1t} \\
z_{2t}
\end{bmatrix} 
\end{align}
with $z_{2t}$ satisfying the reverse autoregression $z_{2t} = \rho_n^{-1} z_{2t+1} - \rho_n^{-1} H_{ \perp n}^{\prime} u_{xt}$, which gives rise to the following 
\begin{align}
z_{2t} = - H_{ \perp n}^{\prime} \sum_{j = 1}^{n - t} \rho_n^{-j} u_{xt + j},  
\end{align}
because $z_{2n} = H_{ \perp n}^{\prime} x_n = 0$. Using the orthogonality condition of $H_n$, we obtain the following expression for the least squares estimator after rotation of the regression space
\begin{align*}
n^{ \frac{ 1 + \alpha }{2} } \left( \hat{A}_n - A  \right) 
&= 
\left( \frac{1}{ n^{ ( 1 + \alpha ) / 2 } } \sum_{t=1}^n u_{0t} z_t^{\prime} \right) \left( \frac{1}{n^{1 + \alpha } } \sum_{t=1}^n z_t z_t^{\prime} \right)^{-1} H_n^{\prime}
\\
&= 
\left( \frac{1}{ n^{ ( 1 + \alpha ) / 2 } } \sum_{t=1}^n u_{0t} z_t^{\prime} \right)
\end{align*}
\end{proof}

\newpage

\bibliographystyle{apalike}
\bibliography{myreferences1}

\begin{thebibliography}{}

\bibitem[Adamek et~al., 2022]{adamek2022local}
Adamek, R., Smeekes, S., and Wilms, I. (2022).
\newblock Local projection inference in high dimensions.
\newblock {\em arXiv preprint arXiv:2209.03218}.

\bibitem[Adrian and Brunnermeier, 2016]{Adrian2016covar}
Adrian, T. and Brunnermeier, M.~K. (2016).
\newblock Covar.
\newblock {\em The American Economic Review}, 106(7):1705.

\bibitem[Agosto et~al., 2016]{agosto2016modeling}
Agosto, A., Cavaliere, G., Kristensen, D., and Rahbek, A. (2016).
\newblock Modeling corporate defaults: Poisson autoregressions with exogenous
  covariates (parx).
\newblock {\em Journal of Empirical Finance}, 38:640--663.

\bibitem[Ahn and Reinsel, 1990]{ahn1990estimation}
Ahn, S.~K. and Reinsel, G.~C. (1990).
\newblock Estimation for partially nonstationary multivariate autoregressive
  models.
\newblock {\em Journal of the American statistical association},
  85(411):813--823.

\bibitem[Andersen and Varneskov, 2021]{andersen2021consistent}
Andersen, T.~G. and Varneskov, R.~T. (2021).
\newblock Consistent inference for predictive regressions in persistent
  economic systems.
\newblock {\em Journal of Econometrics}, 224(1):215--244.

\bibitem[Anufriev and Panchenko, 2015]{anufriev2015connecting}
Anufriev, M. and Panchenko, V. (2015).
\newblock Connecting the dots: Econometric methods for uncovering networks with
  an application to the australian financial institutions.
\newblock {\em Journal of Banking \& Finance}, 61:S241--S255.

\bibitem[Armillotta and Fokianos, 2022]{armillotta2022nonlinear}
Armillotta, M. and Fokianos, K. (2022).
\newblock Nonlinear network autoregression.
\newblock {\em arXiv preprint arXiv:2202.03852}.

\bibitem[Badev, 2021]{badev2021nash}
Badev, A. (2021).
\newblock Nash equilibria on (un) stable networks.
\newblock {\em Econometrica}, 89(3):1179--1206.

\bibitem[Barigozzi et~al., 2023]{barigozzi2023fnets}
Barigozzi, M., Cho, H., and Owens, D. (2023).
\newblock Fnets: Factor-adjusted network estimation and forecasting for
  high-dimensional time series.
\newblock {\em Journal of Business \& Economic Statistics}, pages 1--13.

\bibitem[Barun{\'\i}k and K{\v{r}}ehl{\'\i}k, 2018]{barunik2018measuring}
Barun{\'\i}k, J. and K{\v{r}}ehl{\'\i}k, T. (2018).
\newblock Measuring the frequency dynamics of financial connectedness and
  systemic risk.
\newblock {\em Journal of Financial Econometrics}, 16(2):271--296.

\bibitem[Basu and Subba~Rao, 2023]{basu2023graphical}
Basu, S. and Subba~Rao, S. (2023).
\newblock Graphical models for nonstationary time series.
\newblock {\em The Annals of Statistics}, 51(4):1453--1483.

\bibitem[Breitung and Demetrescu, 2015]{breitung2015instrumental}
Breitung, J. and Demetrescu, M. (2015).
\newblock Instrumental variable and variable addition based inference in
  predictive regressions.
\newblock {\em Journal of Econometrics}, 187(1):358--375.

\bibitem[Bykhovskaya, 2022]{bykhovskaya2022time}
Bykhovskaya, A. (2022).
\newblock Time series approach to the evolution of networks: Prediction and
  estimation.
\newblock {\em Journal of Business \& Economic Statistics}, 41(1):170--183.

\bibitem[Cavanagh et~al., 1995]{cavanagh1995inference}
Cavanagh, C.~L., Elliott, G., and Stock, J.~H. (1995).
\newblock Inference in models with nearly integrated regressors.
\newblock {\em Econometric theory}, 11(5):1131--1147.

\bibitem[Chen et~al., 2019]{chen2019tail}
Chen, C. Y.-H., H{\"a}rdle, W.~K., and Okhrin, Y. (2019).
\newblock Tail event driven networks of sifis.
\newblock {\em Journal of Econometrics}, 208(1):282--298.

\bibitem[Chen et~al., 2023]{chen2023community}
Chen, E.~Y., Fan, J., and Zhu, X. (2023).
\newblock Community network auto-regression for high-dimensional time series.
\newblock {\em Journal of Econometrics}, 235(2):1239--1256.

\bibitem[Cho et~al., 2023]{cho2023high}
Cho, H., Maeng, H., Eckley, I.~A., and Fearnhead, P. (2023).
\newblock High-dimensional time series segmentation via factor-adjusted vector
  autoregressive modeling.
\newblock {\em Journal of the American Statistical Association}, pages 1--13.

\bibitem[Cubadda and Hecq, 2022]{cubadda2022dimension}
Cubadda, G. and Hecq, A. (2022).
\newblock Dimension reduction for high-dimensional vector autoregressive
  models.
\newblock {\em Oxford Bulletin of Economics and Statistics}, 84(5):1123--1152.

\bibitem[Daouia et~al., 2022]{daouia2022extremile}
Daouia, A., Gijbels, I., and Stupfler, G. (2022).
\newblock Extremile regression.
\newblock {\em Journal of the American Statistical Association},
  117(539):1579--1586.

\bibitem[Davis et~al., 2016]{davis2016asymptotic}
Davis, R.~A., Mikosch, T., and Pfaffel, O. (2016).
\newblock Asymptotic theory for the sample covariance matrix of a heavy-tailed
  multivariate time series.
\newblock {\em Stochastic Processes and their Applications}, 126(3):767--799.

\bibitem[Diebold and Y{\i}lmaz, 2014]{diebold2014network}
Diebold, F.~X. and Y{\i}lmaz, K. (2014).
\newblock On the network topology of variance decompositions: Measuring the
  connectedness of financial firms.
\newblock {\em Journal of Econometrics}, 182(1):119--134.

\bibitem[Dou and M{\"u}ller, 2021]{dou2021generalized}
Dou, L. and M{\"u}ller, U.~K. (2021).
\newblock Generalized local-to-unity models.
\newblock {\em Econometrica}, 89(4):1825--1854.

\bibitem[Doukhan et~al., 2023]{doukhan2023stationarity}
Doukhan, P., Neumann, M.~H., and Truquet, L. (2023).
\newblock Stationarity and ergodic properties for some observation-driven
  models in random environments.
\newblock {\em The Annals of Applied Probability}, 33(6B):5145--5170.

\bibitem[Fan et~al., 2023]{fan2023estimation}
Fan, Y., Han, F., and Park, H. (2023).
\newblock Estimation and inference in a high-dimensional semiparametric
  gaussian copula vector autoregressive model.
\newblock {\em Journal of Econometrics}, 237(1):105513.

\bibitem[Fang et~al., 2023]{fang2023determination}
Fang, P., Gao, Z., and Tsay, R.~S. (2023).
\newblock Determination of the effective cointegration rank in high-dimensional
  time-series predictive regressions.
\newblock {\em arXiv preprint arXiv:2304.12134}.

\bibitem[Fort and Roberts, 2005]{fort2005subgeometric}
Fort, G. and Roberts, G.~O. (2005).
\newblock Subgeometric ergodicity of strong markov processes.
\newblock {\em Annals of Applied Probability}, 15(2):1565--1589.

\bibitem[Gobet and Matulewicz, 2017]{gobet2017parameter}
Gobet, E. and Matulewicz, G. (2017).
\newblock Parameter estimation of ornstein--uhlenbeck process generating a
  stochastic graph.
\newblock {\em Statistical Inference for Stochastic Processes}, 20:211--235.

\bibitem[H{\"a}rdle et~al., 2016]{hardle2016tenet}
H{\"a}rdle, W.~K., Wang, W., and Yu, L. (2016).
\newblock Tenet: Tail-event driven network risk.
\newblock {\em Journal of Econometrics}, 192(2):499--513.

\bibitem[He and Song, 2018]{he2018measuring}
He, X. and Song, K. (2018).
\newblock Measuring diffusion over a large network.
\newblock {\em arXiv preprint arXiv:1812.04195}.

\bibitem[Holberg and Ditlevsen, 2023]{holberg2023uniform}
Holberg, C. and Ditlevsen, S. (2023).
\newblock Uniform inference for cointegrated vector autoregressive processes.
\newblock {\em arXiv preprint arXiv:2306.03632}.

\bibitem[Huang et~al., 2020]{huang2020two}
Huang, D., Wang, F., Zhu, X., and Wang, H. (2020).
\newblock Two-mode network autoregressive model for large-scale networks.
\newblock {\em Journal of Econometrics}, 216(1):203--219.

\bibitem[Jansson and Moreira, 2006]{jansson2006optimal}
Jansson, M. and Moreira, M.~J. (2006).
\newblock Optimal inference in regression models with nearly integrated
  regressors.
\newblock {\em Econometrica}, 74(3):681--714.

\bibitem[Kapetanios et~al., 2014]{kapetanios2014nonlinear}
Kapetanios, G., Mitchell, J., and Shin, Y. (2014).
\newblock A nonlinear panel data model of cross-sectional dependence.
\newblock {\em Journal of Econometrics}, 179(2):134--157.

\bibitem[Kasparis et~al., 2015]{kasparis2015nonparametric}
Kasparis, I., Andreou, E., and Phillips, P. C.~B. (2015).
\newblock Nonparametric predictive regression.
\newblock {\em Journal of Econometrics}, 185(2):468--494.

\bibitem[Katsouris, 2021]{katsouris2021optimal}
Katsouris, C. (2021).
\newblock Optimal portfolio choice and stock centrality for tail risk events.
\newblock {\em arXiv preprint arXiv:2112.12031}.

\bibitem[Katsouris, 2023a]{katsouris2023estimating}
Katsouris, C. (2023a).
\newblock Estimating conditional value-at-risk with nonstationary quantile
  predictive regression models.
\newblock {\em arXiv preprint arXiv:2311.08218}.

\bibitem[Katsouris, 2023b]{katsouris2023limit}
Katsouris, C. (2023b).
\newblock Limit theory under network dependence and nonstationarity.
\newblock {\em arXiv preprint arXiv:2308.01418}.

\bibitem[Katsouris, 2023c]{katsouris2023optimal}
Katsouris, C. (2023c).
\newblock Optimal estimation methodologies for panel data regression models.
\newblock {\em arXiv preprint arXiv:2311.03471}.

\bibitem[Katsouris, 2023d]{katsouris2023statistical}
Katsouris, C. (2023d).
\newblock Statistical estimation for covariance structures with tail estimates
  using nodewise quantile predictive regression models.
\newblock {\em arXiv preprint arXiv:2305.11282}.

\bibitem[Keeling and Eames, 2005]{keeling2005networks}
Keeling, M.~J. and Eames, K.~T. (2005).
\newblock Networks and epidemic models.
\newblock {\em Journal of the royal society interface}, 2(4):295--307.

\bibitem[Kojevnikov et~al., 2021]{kojevnikov2021limit}
Kojevnikov, D., Marmer, V., and Song, K. (2021).
\newblock Limit theorems for network dependent random variables.
\newblock {\em Journal of Econometrics}, 222(2):882--908.

\bibitem[Kostakis et~al., 2015]{kostakis2015Robust}
Kostakis, A., Magdalinos, T., and Stamatogiannis, M.~P. (2015).
\newblock Robust econometric inference for stock return predictability.
\newblock {\em The Review of Financial Studies}, 28(5):1506--1553.

\bibitem[Kostakis et~al., 2018]{kostakis2018taking}
Kostakis, A., Magdalinos, T., and Stamatogiannis, M.~P. (2018).
\newblock Taking stock of long-horizon predictability tests: Are factor returns
  predictable?
\newblock {\em Available at SSRN 3284149}.

\bibitem[Krampe et~al., 2023]{krampe2023structural}
Krampe, J., Paparoditis, E., and Trenkler, C. (2023).
\newblock Structural inference in sparse high-dimensional vector
  autoregressions.
\newblock {\em Journal of Econometrics}, 234(1):276--300.

\bibitem[Laurent and Shi, 2022]{laurent2022unit}
Laurent, S. and Shi, S. (2022).
\newblock Unit root test with high-frequency data.
\newblock {\em Econometric Theory}, 38(1):113--171.

\bibitem[Lee, 2016]{lee2016predictive}
Lee, J.~H. (2016).
\newblock Predictive quantile regression with persistent covariates: Ivx-qr
  approach.
\newblock {\em Journal of Econometrics}, 192(1):105--118.

\bibitem[Lee and Song, 2019]{lee2019stable}
Lee, J.~H. and Song, K. (2019).
\newblock Stable limit theorems for empirical processes under conditional
  neighborhood dependence.
\newblock {\em Bernoulli}, 25(2):1189--1224.

\bibitem[Liu and Phillips, 2023]{liu2023robust}
Liu, Y. and Phillips, P.~C. (2023).
\newblock Robust inference with stochastic local unit root regressors in
  predictive regressions.
\newblock {\em Journal of Econometrics}, 235(2):563--591.

\bibitem[Magdalinos, 2021]{magdalinos2021least}
Magdalinos, T. (2021).
\newblock Least squares and ivx limit theory in systems of predictive
  regressions with garch innovations.
\newblock {\em Econometric Theory}, pages 1--38.

\bibitem[Magdalinos and Phillips, 2009a]{Magdal2009limit}
Magdalinos, T. and Phillips, P. C.~B. (2009a).
\newblock Limit theory for cointegrated systems with moderately integrated and
  moderately explosive regressors.
\newblock {\em Econometric Theory}, 25(2):482--526.

\bibitem[Magdalinos and Phillips, 2009b]{magdalinos2009limit}
Magdalinos, T. and Phillips, P. C.~B. (2009b).
\newblock Limit theory for cointegrated systems with moderately integrated and
  moderately explosive regressors.
\newblock {\em Econometric Theory}, 25(2):482--526.

\bibitem[Maillet et~al., 2015]{maillet2015global}
Maillet, B., Tokpavi, S., and Vaucher, B. (2015).
\newblock Global minimum variance portfolio optimisation under some model risk:
  A robust regression-based approach.
\newblock {\em European Journal of Operational Research}, 244(1):289--299.

\bibitem[Matsui and Pedersen, 2022]{matsui2022characterization}
Matsui, M. and Pedersen, R.~S. (2022).
\newblock Characterization of the tail behavior of a class of bekk processes: A
  stochastic recurrence equation approach.
\newblock {\em Econometric Theory}, 38(1):1--34.

\bibitem[Meitz and Saikkonen, 2008a]{meitz2008ergodicity}
Meitz, M. and Saikkonen, P. (2008a).
\newblock Ergodicity, mixing, and existence of moments of a class of markov
  models with applications to garch and acd models.
\newblock {\em Econometric Theory}, 24(5):1291--1320.

\bibitem[Meitz and Saikkonen, 2008b]{meitz2008stability}
Meitz, M. and Saikkonen, P. (2008b).
\newblock Stability of nonlinear ar-garch models.
\newblock {\em Journal of Time Series Analysis}, 29(3):453--475.

\bibitem[Meitz and Saikkonen, 2021]{meitz2021subgeometric}
Meitz, M. and Saikkonen, P. (2021).
\newblock Subgeometric ergodicity and $\beta$-mixing.
\newblock {\em Journal of Applied Probability}, 58(3):594--608.

\bibitem[Meitz and Saikkonen, 2022]{meitz2022subgeometrically}
Meitz, M. and Saikkonen, P. (2022).
\newblock Subgeometrically ergodic autoregressions with autoregressive
  conditional heteroskedasticity.
\newblock {\em arXiv preprint arXiv:2205.11953}.

\bibitem[Menzel, 2021]{menzel2021bootstrap}
Menzel, K. (2021).
\newblock Bootstrap with cluster-dependence in two or more dimensions.
\newblock {\em Econometrica}, 89(5):2143--2188.

\bibitem[Mikusheva, 2007]{mikusheva2007uniform}
Mikusheva, A. (2007).
\newblock Uniform inference in autoregressive models.
\newblock {\em Econometrica}, 75(5):1411--1452.

\bibitem[Mitchener and Richardson, 2019]{mitchener2019network}
Mitchener, K.~J. and Richardson, G. (2019).
\newblock Network contagion and interbank amplification during the great
  depression.
\newblock {\em Journal of Political Economy}, 127(2):000--000.

\bibitem[Nicholson et~al., 2017]{nicholson2017varx}
Nicholson, W.~B., Matteson, D.~S., and Bien, J. (2017).
\newblock Varx-l: Structured regularization for large vector autoregressions
  with exogenous variables.
\newblock {\em International Journal of Forecasting}, 33(3):627--651.

\bibitem[Olmo and Sanso-Navarro, 2023]{olmo2023nonparametric}
Olmo, J. and Sanso-Navarro, M. (2023).
\newblock A nonparametric spatial regression model using partitioning
  estimators.
\newblock {\em Econometrics and Statistics}.

\bibitem[Paruolo, 1997]{paruolo1997asymptotic}
Paruolo, P. (1997).
\newblock Asymptotic inference on the moving average impact matrix in
  cointegrated 1 (1) var systems.
\newblock {\em Econometric Theory}, 13(1):79--118.

\bibitem[Phillips, 1987a]{Phillips1987time}
Phillips, P. C.~B. (1987a).
\newblock Time series regression with a unit root.
\newblock {\em Econometrica: Journal of the Econometric Society}, pages
  277--301.

\bibitem[Phillips, 1987b]{Phillips1987towards}
Phillips, P. C.~B. (1987b).
\newblock Towards a unified asymptotic theory for autoregression.
\newblock {\em Biometrika}, 74(3):535--547.

\bibitem[Phillips, 2014]{phillips2014confidence}
Phillips, P. C.~B. (2014).
\newblock On confidence intervals for autoregressive roots and predictive
  regression.
\newblock {\em Econometrica}, 82(3):1177--1195.

\bibitem[Phillips and Lee, 2013]{phillips2013predictive}
Phillips, P. C.~B. and Lee, J.~H. (2013).
\newblock Predictive regression under various degrees of persistence and robust
  long-horizon regression.
\newblock {\em Journal of Econometrics}, 177(2):250--264.

\bibitem[Phillips and Magdalinos, 2007]{phillips2007limit}
Phillips, P. C.~B. and Magdalinos, T. (2007).
\newblock Limit theory for moderate deviations from a unit root.
\newblock {\em Journal of Econometrics}, 136(1):115--130.

\bibitem[Phillips and Magdalinos, 2008]{Phillips2008limit}
Phillips, P. C.~B. and Magdalinos, T. (2008).
\newblock Limit theory for explosively cointegrated systems.
\newblock {\em Econometric Theory}, 24(4):865--887.

\bibitem[Phillips and Magdalinos, 2009]{PM2009econometric}
Phillips, P. C.~B. and Magdalinos, T. (2009).
\newblock Econometric inference in the vicinity of unity.
\newblock {\em Singapore Management University, CoFie Working Paper}, 7.

\bibitem[Phillips and Magdalinos, 2013]{Phillips2013inconsistent}
Phillips, P. C.~B. and Magdalinos, T. (2013).
\newblock Inconsistent var regression with common explosive roots.
\newblock {\em Econometric Theory}, 29(4):808--837.

\bibitem[Poskitt, 2006]{poskitt2006identification}
Poskitt, D.~S. (2006).
\newblock On the identification and estimation of nonstationary and
  cointegrated armax systems.
\newblock {\em Econometric Theory}, 22(6):1138--1175.

\bibitem[Reinsel, 1983]{reinsel1983some}
Reinsel, G. (1983).
\newblock Some results on multivariate autoregressive index models.
\newblock {\em Biometrika}, 70(1):145--156.

\bibitem[Schennach, 2018]{schennach2018long}
Schennach, S.~M. (2018).
\newblock Long memory via networking.
\newblock {\em Econometrica}, 86(6):2221--2248.

\bibitem[Ter{\"a}svirta et~al., 1994]{terasvirta1994aspects}
Ter{\"a}svirta, T., Tj{\o}stheim, D., and Granger, C.~W. (1994).
\newblock Aspects of modelling nonlinear time series.
\newblock {\em Handbook of econometrics}, 4:2917--2957.

\bibitem[Toda and Yamamoto, 1995]{toda1995statistical}
Toda, H.~Y. and Yamamoto, T. (1995).
\newblock Statistical inference in vector autoregressions with possibly
  integrated processes.
\newblock {\em Journal of econometrics}, 66(1-2):225--250.

\bibitem[White, 2000]{white2000asymptotic}
White, H. (2000).
\newblock {\em Asymptotic Theory for Econometricians}.
\newblock Academic Press.

\bibitem[Zhang, 2023]{zhang2023statistical}
Zhang, Y. (2023).
\newblock Statistical inference of high-dimensional vector autoregressive time
  series with non-iid innovations.
\newblock {\em arXiv preprint arXiv:2310.07364}.

\bibitem[Zhu et~al., 2020]{zhu2020multivariate}
Zhu, X., Huang, D., Pan, R., and Wang, H. (2020).
\newblock Multivariate spatial autoregressive model for large scale social
  networks.
\newblock {\em Journal of Econometrics}, 215(2):591--606.

\bibitem[Zhu and Pan, 2020]{zhu2020grouped}
Zhu, X. and Pan, R. (2020).
\newblock Grouped network vector autoregression.
\newblock {\em Statistica Sinica}, 30(3):1437--1462.

\bibitem[Zhu et~al., 2017]{zhu2017network}
Zhu, X., Pan, R., Li, G., Liu, Y., Wang, H., et~al. (2017).
\newblock Network vector autoregression.
\newblock {\em The Annals of Statistics}, 45(3):1096--1123.

\bibitem[Zhu et~al., 2019]{zhu2019network}
Zhu, X., Wang, W., Wang, H., and H{\"a}rdle, W.~K. (2019).
\newblock Network quantile autoregression.
\newblock {\em Journal of econometrics}, 212(1):345--358.

\end{thebibliography}

\newpage

\end{document}